\documentclass[pdflatex,sn-mathphys-num]{sn-jnl}

\usepackage{graphicx}%
\usepackage{multirow}%
\usepackage{amsmath,amssymb,amsfonts}%
\usepackage{amsthm}%
\usepackage{mathrsfs}%
\usepackage[title]{appendix}%
\usepackage{xcolor}%
\usepackage{textcomp}%
\usepackage{manyfoot}%
\usepackage{booktabs}%
\usepackage{algorithm}%
\usepackage{algorithmicx}%
\usepackage{algpseudocode}%
\usepackage{listings}%

\theoremstyle{thmstyleone}%
\newtheorem{theorem}{Theorem}
\newtheorem{proposition}[theorem]{Proposition}%

\theoremstyle{thmstyletwo}%

\theoremstyle{thmstylethree}%
\newtheorem{definition}{Definition}%

\usepackage{algorithm}

\usepackage{amsmath}
\usepackage{amssymb}
\usepackage{mathtools}
\usepackage{amsthm}
\usepackage{multirow}
 \usepackage{booktabs}
 
\newcommand{\R}{\mathbb{R}}

\newcommand{\K}{\mathcal{K}}

\usepackage{algorithm}
\usepackage{algpseudocode}

\algrenewcommand\algorithmicrequire{\textbf{Input:}}
\algrenewcommand\algorithmicensure{\textbf{Output:}}

\newcommand{\trls}{\textsc{TRLS}}

\usepackage{booktabs}
\usepackage{tcolorbox}

\begin{document}

\title[Topologically Relevant Landmark Selection on Point Clouds]{Topologically Relevant Landmark Selection on Point Clouds}


\author[1,2]{\fnm{Kaifeng} \sur{Zhang}}\email{zhangkf2022@lamda.nju.edu.cn}

\author[1,2]{\fnm{Kai Ming} \sur{Ting}}\email{tingkm@nju.edu.cn}


\affil[1]{\orgdiv{National Key Laboratory for Novel Software Technology}, \orgname{Nanjing University}, \orgaddress{\street{Xianlin Avenue}, \city{Nanjing}, \postcode{210023}, \state{Jiangsu}, \country{China}}}

\affil[2]{\orgdiv{School of Artificial Intelligence}, \orgname{Nanjing University}, \orgaddress{\street{Xianlin Avenue}, \city{Nanjing}, \postcode{210023}, \state{Jiangsu}, \country{China}}}





\abstract{Persistent Homology (PH) is an important tool in Topological Data Analysis for point clouds, but can be prohibitively expensive for large datasets. A practical approximation is to compute PH on a smaller subset of representative points, known as landmarks. Existing landmark selection methods mainly address challenges such as outlier contamination rather than restoring the topological features of the full dataset. To select topologically relevant landmarks, we propose a landmark selection method based on Hodge Laplacian. It assigns each point a topological relevance score based on its harmonic participation and selects points with high scores as landmarks. By quantifying each point’s participation in global homological structures, our method restores the PH of the original point cloud more accurately than geometric baselines, a graph coarsening approach and a method built upon local PH on both synthetic and real-world datasets.}




\maketitle

\section{Introduction}


Persistent Homology (PH) \citep{wasserman2016topological,edelsbrunner2000topological} provides a principled framework for characterizing the multi-scale topological features of point cloud data. Given a point cloud, PH constructs a nested sequence of simplicial complexes \citep{salnikov2018simplicial}, known as a 
\textit{filtration}, and tracks how topological features appear and disappear along this sequence. The resulting summary is encoded by a Persistence Diagram (PD), denoted as a finite multiset $\mathcal{D}_q=\{r=(b,d) \mid d\geq b\}$, where each point \(r\) represents a topological feature, such as a connected component, circle, or void. 


 PDs can be abstract descriptors of point cloud topology that carry domain-specific information. For instance, PDs computed from atomic point clouds have been used to analyze local and intermediate scale structural order in amorphous materials \citep{hiraoka2016hierarchical}. In addition, PDs have been shown to provide useful structural descriptors in a wide range of scientific domains, such as biology \citep{xia2015multidimensional,kovacev2016using,xia2018multiscale,meng2020weighted,liu2022dowker} and chemistry \citep{doi:10.1073/pnas.1520877113,lee2017quantifying,townsend2020representation}. 

Despite the effectiveness of PDs in characterizing point cloud topology, their computation can become prohibitively expensive for large point clouds. As pointed by \citet{zomorodian2004computing}, the cost of PH grows rapidly with the number of simplices, making direct
computation prohibitive for large point clouds.  This difficulty motivates the use of subsampling strategies that compute PH on a smaller set of representative points, known as landmarks \citep{de2004topological,stolz2023outlier}.


Existing landmark selection methods include random sampling, farthest point sampling, and strategies based on local PH \citep{stolz2023outlier} and graph coarsening \citep{topograph}. Random sampling selects landmarks uniformly from the input point cloud, making it simple and efficient. Farthest Point Sampling (FPS) \citep{eldar1997farthest,ebli2019notion,chen2021decomposition,climaco2026densityaware} improves spatial coverage by iteratively selecting points that are far from the current landmark set, and hence tends to produce a more evenly distributed subset. However, this distance-based criterion can be sensitive to outliers \citep{de2004topological}. \citet{stolz2023outlier} proposed a landmark selection method that uses local PH to reduce the influence of outliers. And \citet{topograph} proposed TopoGraph, a fast method which coarsens the graph built on the point cloud. Despite the development of these methods, a central question remains unaddressed:

\begin{tcolorbox}[colback=gray!8, colframe=gray!40]
\emph{Can PH computed on selected landmarks faithfully restore the underlying clean topological features of a contaminated point cloud?}
\end{tcolorbox}
This question is critical because good coverage and few outliers do not ensure
reliable downstream analysis unless the landmarks recover the underlying clean topological features.

Existing methods focus on different objectives. Random sampling targets
unbiased subsampling and FPS targets geometric coverage. The PH landmarks of \citet{stolz2023outlier} use persistent homology in fixed-radius neighborhoods to select landmarks robustly in the presence of outliers, aiming to retain the coarse topology of the underlying signal. Their experiments primarily assess the fraction of selected landmarks that belong to the signal. This fraction, however, does not determine whether the selected set preserves the signal's global PD: even landmarks drawn entirely from the signal can miss a loop or void, or alter its persistence, if they provide insufficient coverage. Likewise, a pointwise score based on local PH does not directly relate to the topology (PD) of the selected landmarks. TopoGraph \cite{topograph} does spectral graph coarsening \cite{Deng2020GraphZoom,10.24963/ijcai.2024/891} on the graph built on outlier-filtered point cloud. The nodes of the coarsened graph are the then used as landmarks. \emph{None of
these methods considers landmark selection from the perspective of global topology. Thus,
they do not answer the question above.}

To address this problem, we propose a landmark selection method based on global Hodge Laplacian \citep{lim2020hodge,schaub2020random}, which is shown by \citet{tpcc,grande2025pointlevel} to be able to effectively encode the global topological features in PH. Our method assigns each point a topological relevance score derived from a Hodge Laplacian built on point cloud, and selects the points with high scores as landmarks. These landmarks are therefore encouraged to retain points that are important for the global topological structure of the underlying point cloud. Our contributions are fourfold:
\begin{enumerate}

\item Establish topology recovery as the central objective of landmark selection and formalize this objective as the \emph{Topology Recovery Task}: selecting
landmarks whose PH recovers the clean topological features underlying an
outlier-contaminated point cloud.
    \item Propose Topologically Relevant Landmark Selection (TRLS) method, which combines density-guided complex construction with Hodge harmonic participation for Topology Recovery Task.
    \item Conduct experiments on both synthetic and real-world datasets, showing that our method often outperforms random sampling, geometric baselines (FPS and its density-aware variant \cite{climaco2026densityaware}), graph coarsening approach TopoGraph \cite{topograph} and a competitor based on local PH \cite{stolz2023outlier} in topology recovery task. 
\end{enumerate}

\section{Background on Persistent Homology}
 Persistent Homology (PH) \citep{wasserman2016topological,edelsbrunner2000topological} studies the evolution of topology across scales by constructing a filtration of simplicial complexes. A common construction is the Vietoris Rips filtration \citep{vr1995}, where simplices are added according to pairwise distances between points. 

\begin{figure}[h]
    \centering
\includegraphics[width=.99\linewidth]{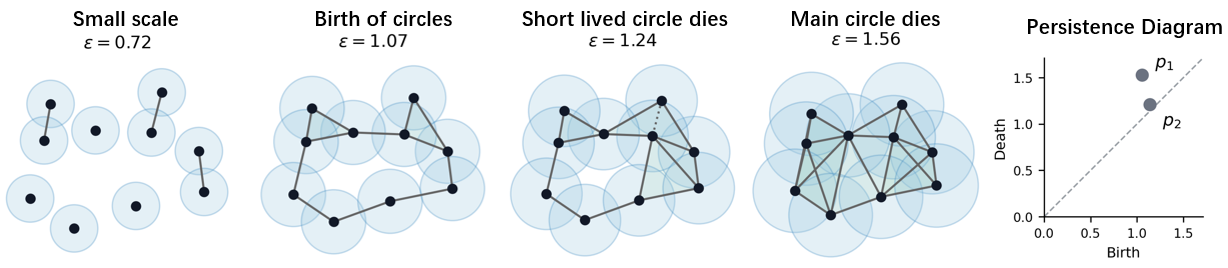}
    \caption{
    Illustration of PH on a point cloud. As the filtration scale \(\epsilon\) increases, rings are born and later die when they are filled by higher dimensional simplices. Persistence Diagram records each circle as a point \(p=(b,d)\). \(p_1=(1.06,1.53)\) denotes the main persistent circle, and \(p_2=(1.14,1.21)\) denotes the short lived one.
    }
    \label{fig: ph background}
\end{figure}



At scale $\epsilon$, $\mathcal{K}(\epsilon)$ is a simplicial complex consisting of vertices, edges, triangles, and higher-dimensional simplices. Each simplex $\sigma\in\mathcal{K}(\epsilon)$ has dimension $\dim(\sigma)=|\sigma|-1$. The $q$-th Betti number is $\beta_q=\operatorname{rank}H_q(\mathcal{K}(\epsilon);\mathbb{F})$, where $H_q$ is the $q$-th homology group over a field $\mathbb{F}$. It counts independent $q$-dimensional topological features. For example, the complex at $\epsilon=1.07$ in Figure~\ref{fig: ph background} contains one circle, giving $\beta_1=1$.

As the filtration scale $\epsilon$ increases, edges between points having distance less than $\epsilon$ are added. And topological features such as connected components, circles appear and then disappear. The homology dimension is
denoted by \(q\): \(q=0\) for connected components, \(q=1\) for loops
or circles, and \(q=2\) for voids. Each feature is represented by its birth and death times in a Persistence Diagram (PD), and its persistence, defined as the difference between death and birth, measures the prominence of the feature. Features with large persistence are often interpreted as meaningful topological features of the underlying space, while short lived features are typically treated as noise. An example of filtration and PD ($q=1$) is given in Figure \ref{fig: ph background}. 


\section{Topology Recovery Task}

Let $X=\{x_i\}_{i=1}^{n}\subset \R^d$ be a point cloud with outliers. In our work, we have $X=X_{\mathcal{M}}\cup X_{\mathrm{o}}$, where $X_{\mathcal{M}}$ are the points sampled from underlying manifold $\mathcal{M}$ and $X_{\mathrm{o}}$ are the outlier points.
We start with the definition of Birth Multiset. 

\begin{definition}[Birth Multiset]
    The birth multiset of a PD $\mathcal{D}_q (X)$ is defined as $\mathcal{B}_q (X)=\{b\mid (b,d)\in \mathcal{D}_q(X), d-b\geq\epsilon\}$, i.e., the multiset containing all the prominent birth scales \footnote{Birth scales of the topological features with extremely small persistence, i.e., noises, are not considered here.} in $\mathcal{D}_q (X)$.
\end{definition}
We then define topological relevance (Definition \ref{def:topological-relevance}) through birth scales because they capture when topological features first become supported by the point cloud. This directly reflects the points responsible for forming circles or voids. Death scales, however, depend on how features are later filled or destroyed and can be influenced by nearby trivial points. Therefore, the birth multiset better matches our goal of identifying points that support topological feature formation, while full PDs remain useful for evaluation in Topological Recovery Task (Definition \ref{definition:TRP}).
\begin{definition}[Topological Relevance]
    Given points $X_\mathcal{M}$ sampled from manifold $\mathcal{M}$ and non-empty homology dimension range $Q=\{q_1,q_2,...\}$, for $q\in Q$, topologically irrelevant points $X_\mathrm{ir}^q$ are the maximal subset $\hat{X}\subset X_\mathcal{M}$ such that $\mathcal{B}_{q}(X_\mathcal{M})=\mathcal{B}_q (X_{\mathcal{M}}\setminus\hat{X})$. And the corresponding topologically relevant points are $X_\mathrm{r}^q=X_\mathcal{M} \setminus X_\mathrm{ir}^q$, i.e., $\mathcal{B}_q(X_\mathcal{M})=\mathcal{B}_q(X_{\mathrm{r}}^q)$. Combining all homology dimensions, we have the (aggregated) topologically relevant points $X_\mathrm{r}=\bigcup_{q\in Q}X_\mathrm{r}^q$ and topologically irrelevant points  $X_\mathrm{ir}=\bigcap_{q\in Q} X_{\mathrm{ir}}^{q}$.
    \label{def:topological-relevance}
\end{definition}
  When $q=0$, the removal of any point $x\in X_\mathcal{M}$ will change the cardinality of $\mathcal{B}_q(X_\mathcal{M})$. Then we have $X_\mathcal{M}=X_\mathrm{r}$. So we only consider $\min_{q\in Q}q\geq 1$ in this work. In practice, each $x\in X_\mathcal{M}$ has perturbation caused by (Gaussian) noise.
\begin{proposition}
Under Definition~\ref{def:topological-relevance}, the topologically relevant and irrelevant points form a decomposition of \(X_{\mathcal M}\), i.e., $X_{\mathcal M}=X_{\mathrm{r}}\cup X_{\mathrm{ir}}, X_{\mathrm{r}}\cap X_{\mathrm{ir}}=\varnothing.$
\label{proposition:topological-relevance}
\end{proposition}


Proposition \ref{proposition:topological-relevance} shows the completeness of the relevance definition. Its proof is given in Appendix \ref{appendix: proof}. An example is given in Figure \ref{fig: demo sphere in circle}. In the \textsc{sphere-in-circle} dataset, when $Q=\{1\}$ ($\{2\}$), $X_\mathrm{r}$ are the points in the circle (sphere). When $Q=\{1,2\}$, $X_\mathrm{r}$ are the points in the circle and sphere. $X_\mathrm{ir}$ are the points on the line connector. The removal of the line connector does not change the 1-dimensional and 2-dimensional PD of the \textsc{sphere-in-circle} dataset. For more synthetic datasets in Figure \ref{fig:three_dataset_comparison}, relevant regions are known by construction: circles, spheres, and torus support the topological features, whereas line connectors or solid spheres do not affect their birth times.

Once we have the notion of topological relevance, we can define Topology Recovery Task (TRT).
\begin{definition}[Topology Recovery Task]
    Given a point cloud
$X=\{x_i\}_{i=1}^n\subset\mathbb R^d$, a sample size $m$, and target
homology dimensions $Q$, Topology Recovery Task asks for selecting a landmark
set $S\subset X$, $|S|=m$, such that the PDs computed on $S$ can serve as
a faithful surrogate for the PDs of the unknown topologically relevant signal
$X_r$.
\label{definition:TRP}
\end{definition}
The selected landmark set is evaluated primarily by the PD error
$\sum_{q\in Q} W_2(D_q(S),D_q(X_r))$, complemented by the signal
fraction $|S\cap X_r|/|S|$, where $W_2$ is the finite
$2$-Wasserstein distance with $L_\infty$ ground metric, a standard
metric for PDs \citep{mileyko2011probability}. $\mathcal{D}_q(X_\mathrm{r})$ and $\mathcal{D}_{q}(S)$ are referred to as signal PD and selected PD, respectively. The consistency between the signal and selected PDs in TRT is crucial when using landmarks as a computational surrogate for the original point cloud in downstream topological analysis. 

A smaller PD error directly indicates more faithful recovery of the
desired clean topology. Signal fraction $|S\cap X_r|/|S|$ is
reported as a diagnostic of contamination and topological irrelevance rejection \cite{stolz2023outlier}. A higher
signal fraction is desired as it means more topologically relevant points are selected. Note that signal fraction alone does not guarantee topology recovery, since the selected
signal points may still fail to cover the relevant topological structures.

Given that $X_r$ is unknown during selection, our goal is to infer a score from
$X$ that assigns high values to points in $X_r$ and low values to irrelevant
points and outliers. We then select points with high scores as landmarks, while $X_r$ is used
only for evaluation.

\section{Related Work}

Subsampling has been used to reduce the computational burden of PH on large point clouds \citep{de2004topological,stolz2023outlier}. Random sampling selects landmarks uniformly from the input point cloud and is computationally efficient, but it may contain topologically irrelevant points and outliers and have limited coverage . 

Farthest Point Sampling (FPS) \citep{de2004topological} selects landmarks sequentially by choosing the point that is farthest from the current landmark set, which usually provides better spatial coverage of the point cloud. Nevertheless, this criterion ignores topological relevance due to its geometric nature and can be sensitive to outliers. As shown in the second row in Figure \ref{fig: demo sphere in circle}, FPS selects many outliers and topologically irrelevant points, i.e., the line connector. Density-Aware FPS (DA-FPS) \citep{climaco2026densityaware} improves FPS by incorporating local density information into the sampling criterion, which can reduce the selection of outliers. However, it remains a primarily geometric strategy, and thus may still favor dense but topologically irrelevant regions.

Several other methods have studied landmark selection (or subsampling) for PH \citep{PDstable2,niyogi2008finding,dufresne2019sampling}. But they generally assume relatively clean point clouds and do not explicitly address the presence of outliers.

In order to improve robustness to outliers, \citet{stolz2023outlier} proposed the first landmark selection method that considers PH. \citet{stolz2023outlier} uses local PH via complex on points in a small local neighborhood \citep{stolz2020geometric,von2023topological} to assign a score to each point, which can be interpreted as measuring how much the local topology around a point deviates from that of a regular point on the underlying manifold. This method is named Persistent Homology Landmarks (PHL). It has two variants: PHL-\MakeUppercase{\romannumeral 1} considers all homology dimensions; PHL-\MakeUppercase{\romannumeral 2} considers homology dimension $q=1$ only. 

PHL is useful for identifying outliers or singular regions. But it considers local topology only and does not directly quantify the contribution of each point to the global topological structure of the point cloud. So the selected landmarks may contain topologically irrelevant points, such as the points from the line connector, as shown in Figure \ref{fig: demo sphere in circle}. 

Topograph \cite{topograph} conducts landmark selection from the view of graph coarsening. It first filters the outliers of the point cloud using DTM (Distance to Measure) \cite{chazal2017robust} values and then builds adaptive k-nearest neighbors graph \cite{cai2022new} on the filtered point cloud. Then it does  spectral graph coarsening \cite{Deng2020GraphZoom,10.24963/ijcai.2024/891}. The nodes of the coarsened graph are the then used as landmarks. TopoGraph focuses on the speed up gained via graph coarsening. It does not consider the global topology of the point cloud in its core methodology.

In summary, existing methods mainly focus on sampling efficiency, spatial coverage, local topological abnormality detection and robustness to outliers. They do not consider global topology of the point cloud. So the landmarks may contain many outliers and topologically irrelevant points. The PD of the landmarks could be very different from that of the topologically relevant points, which goes against the goal of Topological Recovery Task.


\section{Preliminaries on Hodge Laplacian}
To make global topology useful for point-level analysis, we need to connect them to local structures in the point cloud. We start by introducing boundary matrix. Let $\mathcal{K}$ be an oriented simplicial complex \citep{mishra2025foundations} built on a point cloud. We denote  by \(C_q(\mathcal{K})\) the vector space
spanned by $\mathcal{K}_q$, the set of $q$-simplices $\sigma\in \K$ with $|\sigma|=q+1$. The boundary operator
$\partial_q:C_{q+1}(\mathcal{K})\rightarrow C_{q}(\mathcal{K})$ maps each $(q+1)$-simplex to the
alternating sum of its $q$-dimensional faces. After fixing an ordering of
simplices, \(\partial_q\) can be represented by an incidence matrix
\(B_q\in\mathbb{R}^{|\mathcal{K}_{q}|\times |\mathcal{K}_{q+1}|}\), whose entries encode whether a
$q$-dimensional simplex is contained in a $(q+1)$-dimensional simplex. 
Hodge Laplacian is then defined based on boundary matrix as following.

\begin{definition}[\citealp{eckmann1944harmonische,lim2020hodge}]
    The q-th Hodge Laplacian $L_q$ of $\mathcal{K}$ is a square matrix $L_q\in\mathbb{R}^{|\mathcal{K}_{q}|\times |\mathcal{K}_{q}|}$ indexed by the q-simplices of $\mathcal{K}_q$: $        L_q = B_{q-1}^\top B_{q-1} + B_{q}B_{q}^\top$.
\end{definition}

The following key property links the topology of $\mathcal{K}$ with its Hodge Laplacian and enables corresponding simplex-level analysis via eigenvectors. The kernel space $\mathrm{ker}(L_k)$ is also named harmonic space \citep{lim2020hodge,schaub2021signal,roddenberry2021principled}.


\begin{theorem}[\citealp{eckmann1944harmonische,friedman1996computing}]
For a simplicial complex $\mathcal{K}$, let $L_q$ be the Hodge Laplacians
and $\beta_q$ be the Betti numbers of $\mathcal{K}$. We have that
$\mathrm{rank}(\mathrm{ker}(L_q)) = \beta_q$.
\end{theorem}


 \noindent\textbf{Relation with Methods based on Hodge Laplacian.} Like Topological Point Cloud Clustering (TPCC) \citep{tpcc} and Topological Point Features (TOPF) \citep{grande2025pointlevel}, \trls{} uses Hodge Laplacians to relate harmonic structure on simplices to the topology of individual points. Their objectives, however, differ from ours: both of TPCC and TOPF try to cluster points by their membership to the different topological features. TPCC uses the full input cloud's Persistence Diagram to choose the complex scale \citep{tpcc}, builds simplicial complex and use Hodge Laplacian and subspace clustering to learn a representation that is suitable for clustering. TOPF constructs point-level representations after computing Persistent Diagram on the full point cloud and selecting significant generators \citep{grande2025pointlevel}. \trls{} instead selects a fixed-budget landmark set to preserve the clean signal's Persistence Diagram, \textbf{without computing a Persistence Diagram on the full observed point cloud during selection}. Adapting TPCC or TOPF outputs into landmark rankings would change their original tasks while retaining their reported reliance on Persistence Diagram information of full point cloud, so we discuss them as related methods rather than direct landmark selection baselines.

 \section{Landmark Selection via Hodge Laplacian}

For Topology Recovery Task, we propose Topologically Relevant Landmark Selection (TRLS) method by employing Robust Kernel Density Estimation (RKDE) \citep{nips2020robust} for simplicial complex construction and the  harmonic energy of Hodge Laplacian to identify topologically relevant points. The intuition behind \trls{} is that the near-zero eigenspace of Hodge
Laplacian encodes global homological features. Its energy in the subspace on each simplex
measures participation in these features, so selecting points with high energy
preserves the topological structures.

\begin{algorithm}[h]
\caption{Topologically Relevant Landmark Selection}
\label{alg:trls-selection}
\begin{algorithmic}[1]
\Require Point cloud \(X\), sampling density \(\rho\), and target homology dimensions \(Q\)
\Ensure Landmark set \(S\)
\State Fit RKDE to $X$ and compute reliability weight $r_i$
\State Estimate scale parameter $\epsilon$ by the reliability-weighted quantile in Eq.~\eqref{eq:trls-alpha-scale}
\State Construct complex \(\K(X)\) up to dimension \(\max(Q)+1\)
\For{each \(q\in Q\)}
    \State Build Hodge Laplacian $L_q$
    \State Compute simplex energy \(e_q\) by Eq.~\eqref{eq:trls-simplex-energy} using $L_q$
    \State Lift to point-level Hodge-derived score $\hat{v}_q$ by Eq.~\eqref{eq:trls-support-score}
\EndFor
\State Fuse homology dimensions by $    v(i)
    =
    \max_{q\in Q}\hat{v}_q(i)$
\State Set \(m=\lceil\rho n\rceil\)
\State Sort points by Hodge-derived score $v(i)$
\State Select the
$m$ highest-scoring points, randomly breaking exact ties, as landmark set $S$
\State \Return \(S\)
\end{algorithmic}
\end{algorithm}

Given a point cloud \(X=\{x_i\}_{i=1}^{n}\subset\mathbb{R}^d\), a sampling density \(\rho\in(0,1)\), and a set of target homology dimensions \(Q\), \trls{} (Algorithm \ref{alg:trls-selection}) assigns each point a Hodge-derived score and selects \(m=\lceil\rho n\rceil\) landmarks with highest scores. \emph{This procedure does not require PH computation.} 
 \begin{figure*}[t]
    \centering
\includegraphics[width=.99\textwidth]{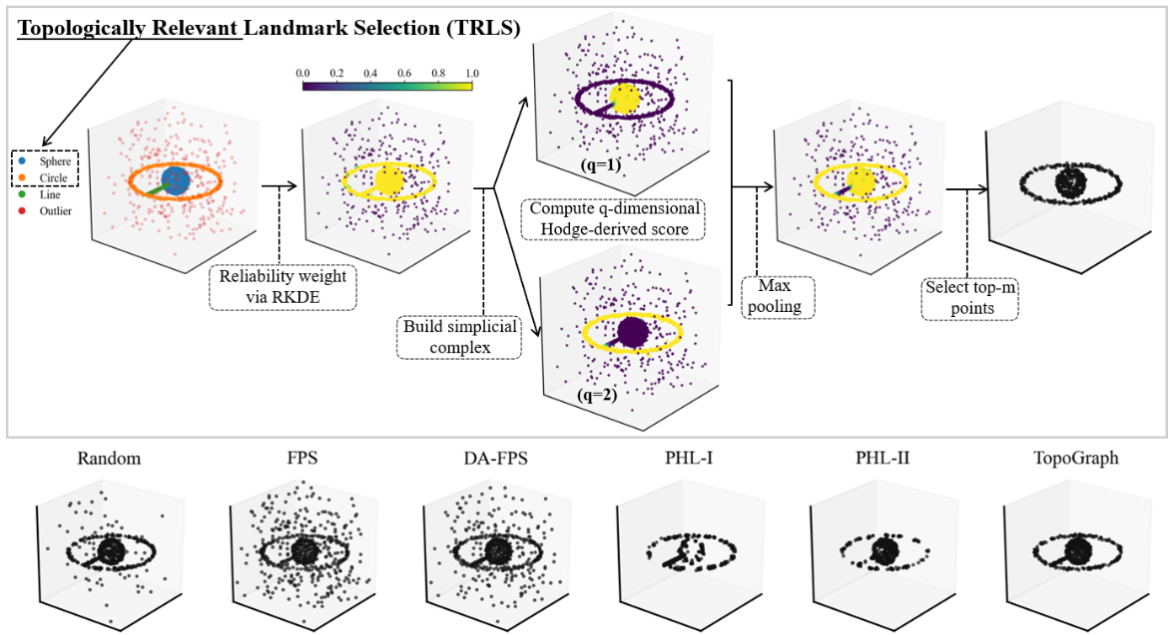}
    \caption{Top row shows the pipeline of TRLS on \textsc{sphere-in-circle} dataset when $Q=\{1,2\}$. This dataset contains points sampled from a circle linked with sphere through a line connector, which are then perturbed with Gaussian noise and contaminated by uniform outliers.  Bottom row shows landmarks selected by competing baselines at sampling density $\rho=0.2$.}
    \label{fig: demo sphere in circle}
\end{figure*}
\subsection{Step 1: Build Simplicial Complex under a Scale Parameter Guided by RKDE}


We first fit RKDE \citep{nips2020robust} on $X$ to reduce the influence of
outliers when constructing the simplicial complex. RKDE robustly estimates the
dominant kernel mean and assigns a reliability weight to each point: points with
large RKHS residuals receive smaller weights. These weights measure each point's
influence on the robust estimate, rather than its local density or a binary
outlier label.

\subsubsection{Obtain Reliability Weights via RKDE}
Let $k_h$ be a Gaussian kernel with bandwidth\footnote{We use a default RKDE bandwidth in this work: the median \(t\) nearest neighbor distances with $t=\operatorname{clip}(\lfloor\log n\rfloor,3,10)$.} $h$. RKDE estimates
$\mu_w=\sum_{j=1}^n w_jk_h(x_j,\cdot)$ by minimizing
$\frac{1}{n}\sum_i\ell(\|k_h(x_i,\cdot)-\mu_w\|_{\mathcal H})$ over
$\Delta_n=\{w\geq0:\sum_iw_i=1\}$, where $\ell$ is the Hampel loss \citep{kim2012robust}. We normalize
$w_i$ to obtain the reliability $r_i\in[0,1]$.


The reliability weight $r_i$ is then obtained by scale $w_i$ to $[0,1]$ through min-max normalization. Note that $r_i$ is not a PH-based quantity. It is a robust inlierness weight used only to choose the scale for building simplicial complex. An example of reliability is shown in the first row of Figure \ref{fig: demo sphere in circle}.

\subsubsection{Choose Scale and Build Simplicial Complex}
Intuitively, the scale parameter \(\epsilon\) \footnote{Our method builds complex at one scale. We focus on point clouds whose relevant topological features appear over a similar range of scales and have comparable local sampling densities. In this setting, the scale chosen using RKDE provides a common complex for Hodge scoring. When target features differ substantially in scale or density, one global scale may fail to represent them together. Constructions that adapt to local density, such as CkNN \citep{2639-8001_2019_1_1} and the filtration based on the $\Lambda$-kernel \citep{pmlrlamdafilter}, may be more suitable. Such heterogeneous regimes are beyond the scope of this work.} is chosen as a robust estimate of the typical local spacing among reliable points, so edges and higher-dimensional simplices are formed at a scale supported by the main point cloud rather than by isolated outliers.

Let \(\delta_i\) be the distance from \(x_i\) to its \(k\)-th nearest neighbor. The scale parameter \footnote{We assume that for each target dimension \(q\in Q\), the RKDE-selected scale \(\epsilon\)
lies inside the alive interval of the target clean features formed by $X_\mathrm{r}$. We validate this in Appendix \ref{appendix: scale validation}.} is chosen by
\begin{equation}
    \epsilon
    = \mathcal{Q}_{\tau}^{(r)}
    \left(\{\delta_i\}_{i=1}^{n}\right),
    \label{eq:trls-alpha-scale}
\end{equation}
where \(\mathcal{Q}_{\tau}^{(r)}\) is the reliability-weighted empirical quantile. 

Here, ``weighted'' means that each local distance \(\delta_i\) contributes according to its reliability \(r_i\), so distances from reliable points have larger influence on the quantile than distances from unreliable or outlying points. The \(\tau\)-quantile is the value below which approximately a \(\tau\) fraction of the total reliability weight falls. In all formal experiments, we use \(k=8\), and \(\tau=0.9\) by default. A sensitivity analysis is given in the Discussion Section.


With the selected scale, we construct the alpha complex \citep{edelsbrunner1983shape,edelsbrunner1994three} $\K$ by retaining simplices with alpha filtration value at most \(\epsilon^2\), up to dimension \(\max(Q)+1\).

Alpha complexes are efficient in 2D or 3D, but rely on Delaunay triangulation, whose size grows rapidly in high dimensions, and require point coordinates. For high dimensional or distance only data, we use Rips complexes instead, as demonstrated in the later experiments.

\subsection{Step 2: Compute Hodge-Derived Score}
\subsubsection{Obtain Simplex-Level Score}
Once we have simplicial complex $\mathcal{K}$, for each \(q\in Q\), we can build the \(q\)-Hodge Laplacian $L_q$. Then we compute the small eigenpairs $(\lambda,u)$ of \(L_q\) under eigenvalue threshold\footnote{$\eta$ is set to be $10^{-8}$ by default. The practical usability of this parameter is discussed in Appendix \ref{appendix: 5}.} $\eta$, where \(\lambda\leq \eta\) is an eigenvalue and
\(u_i\in\mathbb{R}^{|\K_q|}\) is the corresponding eigenvector indexed by
\(q\)-simplices. The simplex-level score is then defined as the harmonic energy\footnote{If no computed eigenvector satisfies the near-harmonic threshold $\eta$, the simplex scores for this dimension are set to zero.} that each \(q\)-simplex \(\sigma\in\K_q\) receives:
\begin{equation}
    e_q(\sigma)
    =
    \sum_{\lambda\le \eta}u(\sigma)^2,
    \label{eq:trls-simplex-energy}
\end{equation}
which is the squared length of the projection of
\(\sigma\) onto the harmonic subspace \(\ker (L_q)\). Thus, the simplex score is a basis-invariant harmonic participation
score. It quantifies the contribution of each simplex ($\sigma$) to the global
\(q\)-dimensional harmonic subspace. This score is then  min-max normalized before lifted to point-level.

\subsubsection{Lift from Simplex-Level to Point-Level}
After we have the simplex-level energy, we need to lift it to point level. Let \(T_q\) be the simplices with \(e_q(\sigma)\) above the \(\tau\)-quantile \footnote{Its sensitivity is analyzed in Appendix \ref{app:simplex-seed-sensitivity}.}, and let $    V_q=\bigcup_{\sigma\in T_q}\operatorname{Vert}(\sigma)$ be the corresponding seed vertices, where $\operatorname{Vert}(\sigma)$ means all the vertices in simplex $\sigma$. Let \(G=(X,\K_1)\) be the one-skeleton graph of the alpha complex, with Euclidean edge weights. For each point \(x_i\), define $   d_q(i)=\min_{v\in V_q}d_G(i,v),$ where \(d_G(i,j)\) is the shortest-path distance between vertices \(i\) and \(j\) on \(G\).
The raw Hodge-derived score under dimension $q$ is defined as
\begin{equation}
    \hat{v}_q(i)
    =
    \exp\left(
    -\left(\frac{d_q(i)}{2\epsilon}\right)^2\right),
    \label{eq:trls-support-score}
\end{equation}
with disconnected vertices assigned score zero. These raw scores are rescaled via min-max normalization to obtain the final Hodge-derived score $\hat{v}_q$ for homology dimension $q$.


\subsection{Step 3: Fusion and Selection}

For the target homology dimensions $Q$, we want to prioritize a point if it is relevant to any one of them. For example, a point on a loop may have a high $\hat{v}_1$ score but a low $\hat{v}_2$ score and averaging the two would lower its rank despite its relevance to the loop. We therefore use max fusion, $v(i)=\max_{q\in Q}\hat v_q(i)$, so that each point is ranked by its strongest response among the target dimensions. Given a sampling density $\rho$, we select the $m$ points with the highest scores and break exact ties at random.

\section{Experiments}
We compare TRLS with random sampling, standard geometric baselines (FPS, DA-FPS), Topograph and methods based on local PH (PHL-\MakeUppercase{\romannumeral 1} \& \MakeUppercase{\romannumeral 2}). PHL-\MakeUppercase{\romannumeral 1} \& \MakeUppercase{\romannumeral 2} and TopoGraph are deterministic methods. For the rest methods involving randomness, we report the mean quantitative metrics (signal fraction \& Wasserstein distance) and stds over five runs. Corresponding stability is analyzed in Appendix \ref{appendix: stability}.

In most datasets, for TRLS and its variants, by default, we use alpha complex for both the landmark selection and evaluation of the subsequent PD. For high-dimensional Grid Cells data \cite{gardner2022toroidal}, Rips complex is used instead. This is intended to show TRLS is agnostic to the complex choice. In addition, we show that different complex choices between score/sample and evaluation stage can give the same best method in Appendix \ref{appendix: 7}. Details on experimental setting and datasets are provided in Appendix \ref{appendix: exp setting}.


\subsection{Validation on Synthetic Datasets}
Our synthetic benchmarks are designed to separate topologically relevant structures from irrelevant parts and outliers. \textsc{Interlocked rings} test H1 preservation, while \textsc{sphere-in-circle} and \textsc{torus-in-circle} contain mixed H1 and H2 signals together with connectors, clutter, and outliers. A summary is given in Table \ref{tab:synthetic-datasets}.
\begin{table}[h]
\centering
\small
\caption{Synthetic benchmark summary. Corresponding visualizations are given in the first column of Figure \ref{fig:three_dataset_comparison}.}
\label{tab:synthetic-datasets}
\setlength{\tabcolsep}{2.4pt}
\begin{tabular}{lcccc}
\toprule
Dataset & Total & Relevant/Irrelevant & Outlier & $Q$ \\
\midrule
\textsc{interlocked rings} & 3000 & 2700/0 & 300 & $\{1\}$ \\
\textsc{sphere-in-circle} & 3000 & 2433/267 & 300 & $\{1,2\}$ \\
\textsc{torus-in-circle} & 4220 & 2400/ 1400 & 420 & $\{1,2\}$ \\
\bottomrule
\end{tabular}
\end{table}

Figure~\ref{fig:three_dataset_comparison} shows that TRLS often selects a higher fraction of signal points while achieving the lowest Wasserstein distance to the signal PD. TopoGraph is the second best method. As shown in Figure \ref{fig:tic_topograph}, the landmarks selected by TopoGraph contains the topologically irrelevant points from line connector and solid spheres on the \textsc{torus-in-circle} dataset because TopoGraph does not consider the homological information. 

\begin{figure}[h]
    \centering
\includegraphics[width=.8\linewidth]{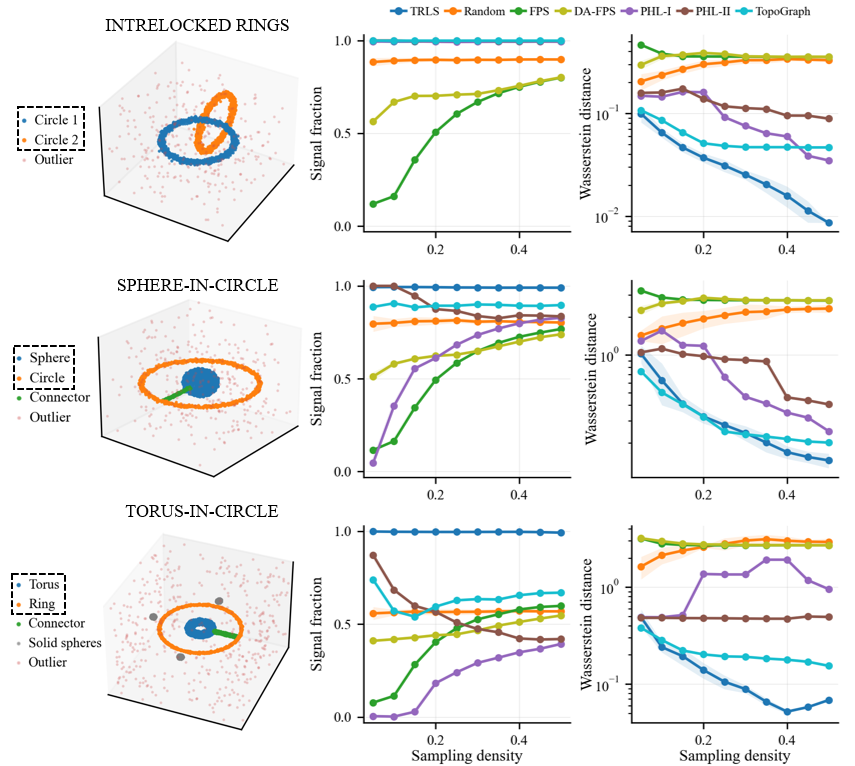}
    \caption{
    Signal fractions and Wasserstein distances ($Q=\{1,2\}$) of different methods over multiple sampling densities on three synthetic point clouds. Labels of topologically relevant points are enclosed by a dashed rectangle.}
\label{fig:three_dataset_comparison}
\end{figure}
PHL-\MakeUppercase{\romannumeral 1} tends to have high scores on points whose neighbors form a Euclidean-like surface, as shown in Figure \ref{fig:tic_score}. But it selects topologically irrelevant points. PHL-\MakeUppercase{\romannumeral 2} only consider $q=1$, so it has high scores of points whose local neighborhood is like 2-dimensional surface. The points in torus are assigned with high score. Hence, it misses the circle and include topologically irrelevant  points from solid spheres.  

In contrast, Figure \ref{fig:tic_score} shows that TRLS explicitly scores points by their contribution to the topological features, leading to samples that better restore signal PDs, as shown in Figure \ref{fig:tic_PD}. 


\begin{figure}[h]
    \centering
\includegraphics[width=.7\linewidth]{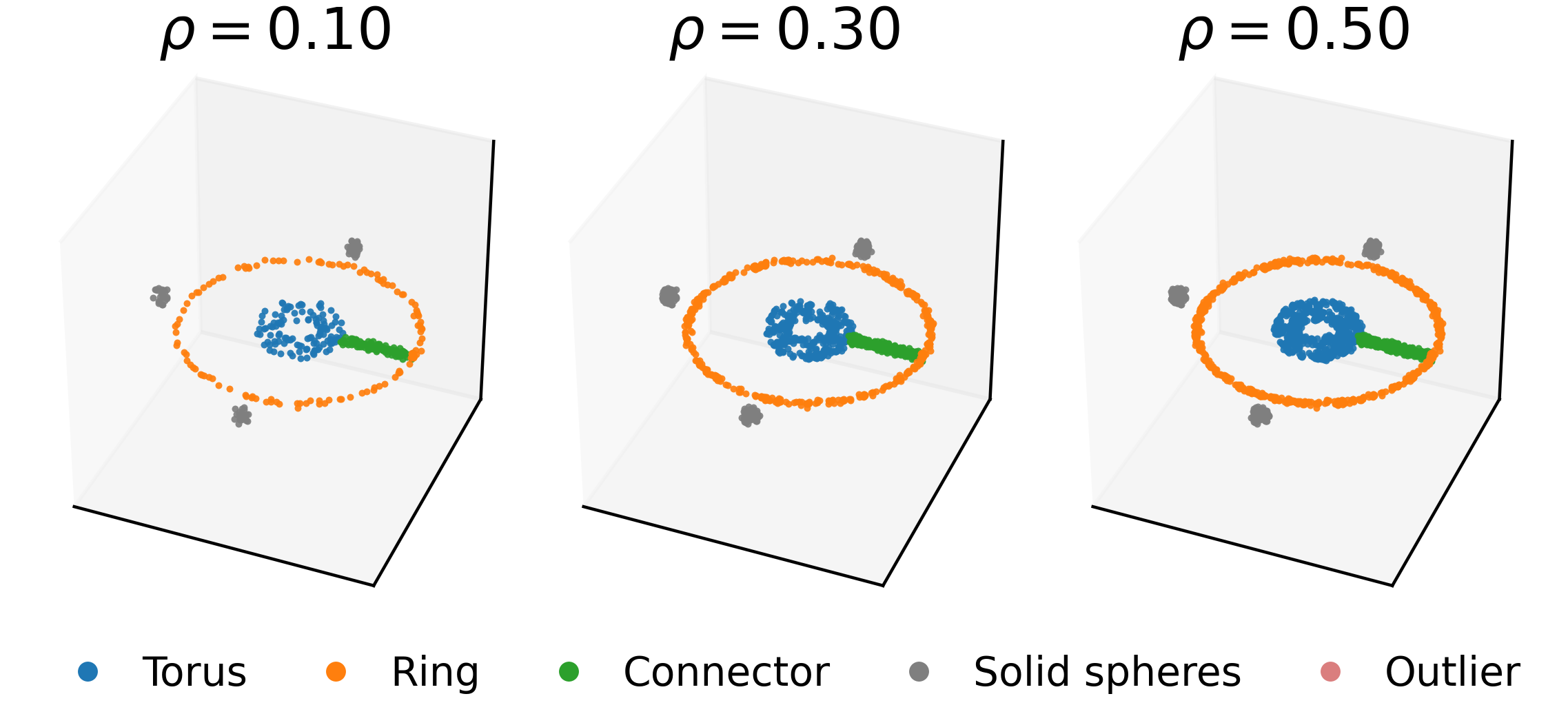}
    \caption{Landmarks selected by TopoGraph under different sampling densities ($\rho$) on the \textsc{torus-in-circle} dataset. 
    }
    \label{fig:tic_topograph}
\end{figure}

\begin{figure}[h]
    \centering
\includegraphics[width=.8\linewidth]{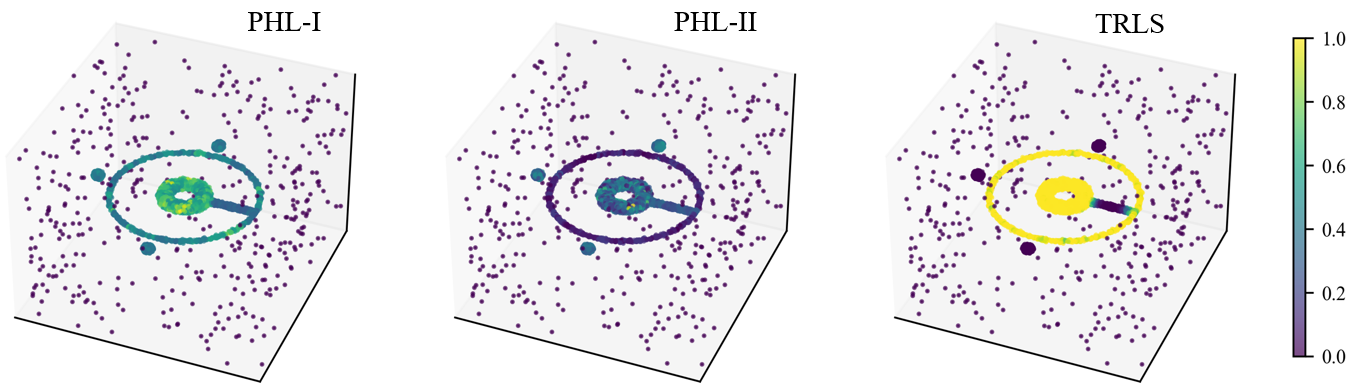}
    \caption{
Normalized scores of PHL-\MakeUppercase{\romannumeral 1}, PHL-\MakeUppercase{\romannumeral 2} and TRLS on \textsc{torus-in-circle} dataset.
    }
    \label{fig:tic_score}
\end{figure}

\begin{figure}[h]
    \centering
\includegraphics[width=.85\linewidth]{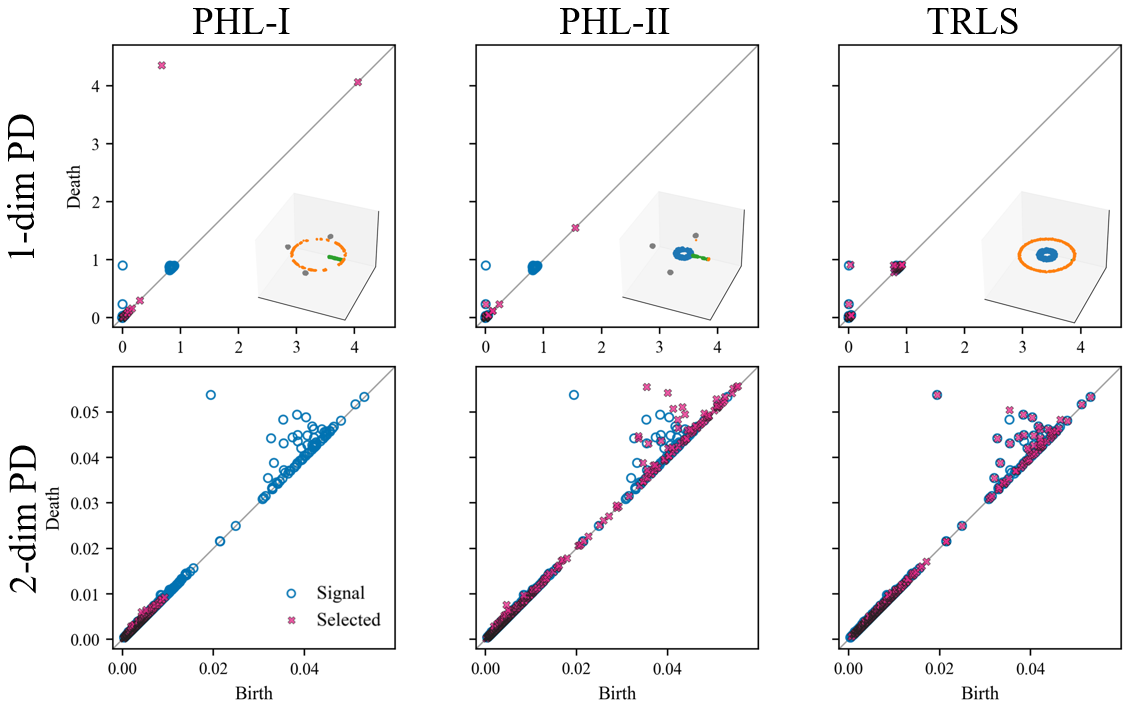}
    \caption{
PDs on the original topologically relevant points and selected landmarks in \textsc{torus-in-circle} dataset selected by PHL-\MakeUppercase{\romannumeral 1}, PHL-\MakeUppercase{\romannumeral 2} and TRLS. The first row also shows the selected landmarks at sampling density $\rho=0.2$.  }
    \label{fig:tic_PD}
\end{figure}

\begin{figure}[h]
    \centering
\includegraphics[width=.92\linewidth]{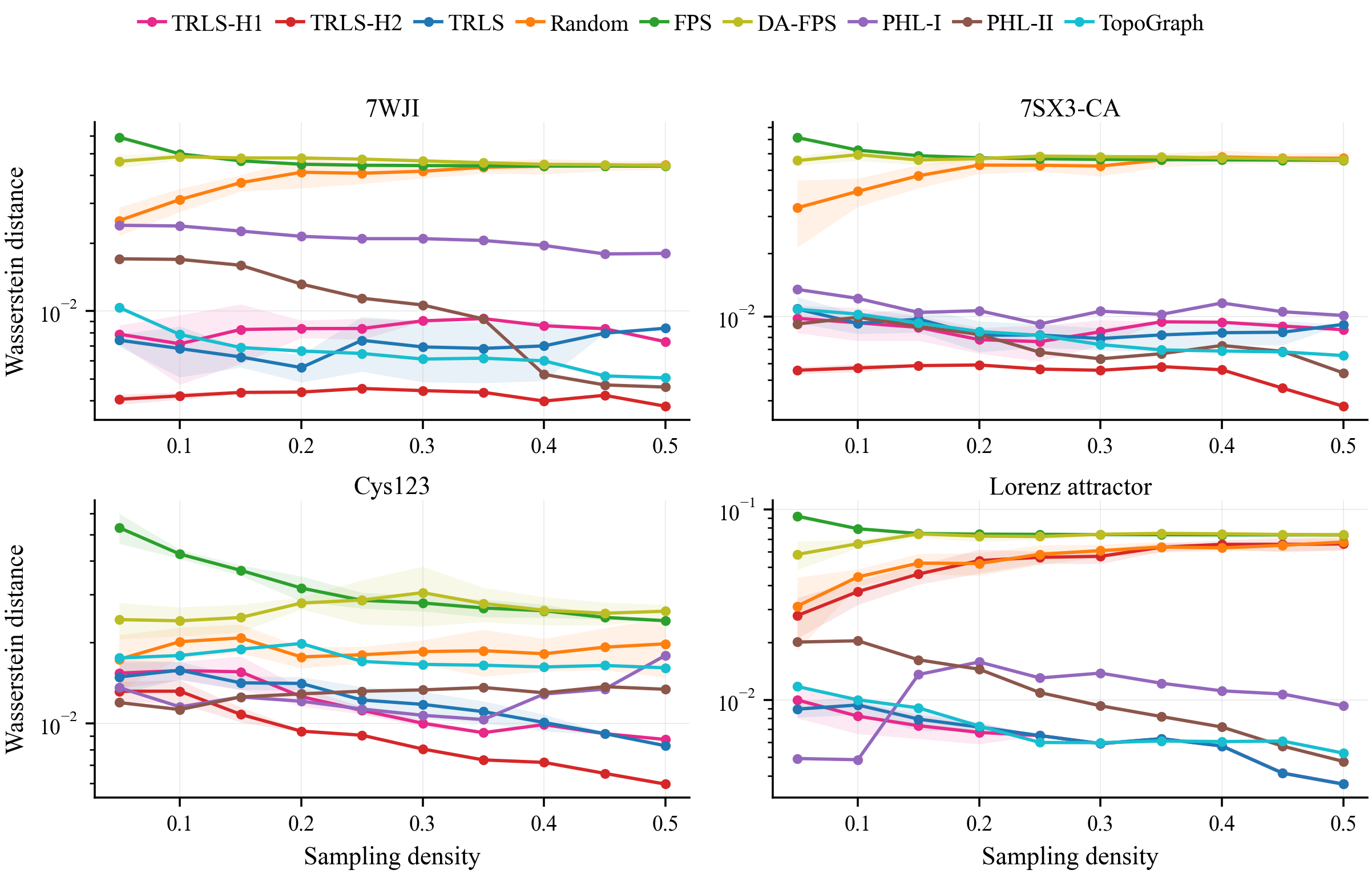}
    \caption{
Wasserstein distances between PDs on protein and dynamic system datasets.
    }
\label{fig:nalcn_wasserstein}
\end{figure}


\subsection{Comparison on Scientific Domain Datasets}
We compare TRLS with baselines on three protein structure point clouds: NALCN channelosome (7WJI\&7SX3-CA) and mutated Cys123 protein, which were also used by \citet{tpcc,grande2025pointlevel}. 7WJI\&7SX3-CA have circles as the main characteristics, reflecting their function as membrane protein. Cys123 has void as main feature as protein pockets, which are crucial for protein-environment interactions. In addition, we use a dynamical system dataset \citep{grande2025pointlevel} containing points sampled in the state space of a Lorenz attractor. Two circles correspond to the two lobes of the attractor. There are no topologically irrelevant points in these four datasets. We add $10\%$ uniform outliers in each point cloud. \textit{For these datasets without point-level relevance labels, we use the original point cloud before outlier injection as the clean signal reference and evaluate topology recovery using Wasserstein distance. Only the outliers are non-signal.}

We compare TRLS with TRLS-H1 and TRLS-H2, which use only Hodge Laplacian $L_1$ and $L_2$, respectively, while the default TRLS uses both for dimension agnostic selection. Evaluation compares the selected PD with the clean signal PD in the Problem relevant dimension: \(q=1\) for NALCN and Lorenz, and \(q=2\) for Cys123, matching their circle and cavity structures. Wasserstein distance results are shown in Figure~\ref{fig:nalcn_wasserstein}. Corresponding results on signal fraction are given in Appendix \ref{appendix: 2}. The results of the TRLS variants on the previous synthetic datasets are analyzed in Appendix \ref{appendix: 1}.

On both NALCN channelosome datasets, TRLS and its variants consistently outperforms Random, FPS, DA-FPS across sampling densities, indicating that topology-aware scores preserve the clean protein circle structure more effectively than geometric or uniform sampling. TRLS-H2 is particularly strong, suggesting that higher-dimensional Hodge information acts as a useful structural regularizer even when the evaluation metric is defined on 1-dimensional PH. PHL baselines and TopoGraph are competitive, but TRLS-H2 attains the smallest Wasserstein distance. 

On Cys123, TRLS-H2 performs the best, consistent with the dataset’s void-like protein pocket structure. This shows that using domain knowledge to choose the target Hodge dimension can improve landmark selection.

On Lorenz attractor dataset, TRLS-H1 and fused TRLS perform best, matching the circle topology of the attractor. This confirms that TRLS can adapt to the relevant dimension, while baselines using geometry alone are less targeted. TopoGraph performs similarly with TRLS on this dataset.

\subsection{Extension on High-dimensional Grid Cells Data}
We use a grid cell population activity dataset \cite{gardner2022toroidal} represented as a point cloud of 10,000 time bins in 544-dimensional neural activity space. We avoid adding outliers here to reduce computational burden. Signal fraction is not reported here because there are no topologically relevance labels for this dataset. Alpha complex cause Out of Memory (OOM) error when computing PD of the entire dataset due to the high dimensionality. So pairwise correlation distances between activity vectors are used as the input metric for Rips complex.

\begin{figure}[h]
    \centering
\includegraphics[width=\linewidth]{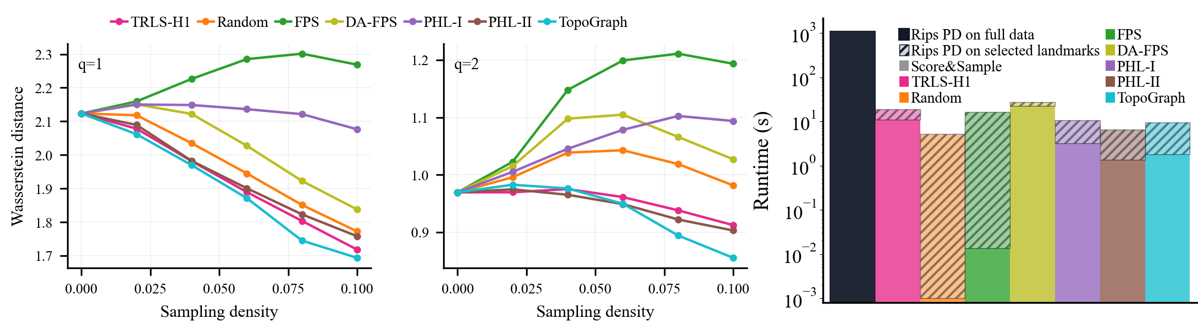}
    \caption{
 The left panel shows Wasserstein distances between Rips PD of the full data and that of selected landmarks with specific homology dimensions $q=1$ and $2$, respectively. The right panel shows runtime comparison when $\rho=0.1$.
    }
\label{fig:grid_cell}
\end{figure}

We compare TRLS-H1 with all the baselines in Figure \ref{fig:grid_cell}. We only consider TRLS-H1 because there are more than 30 million tetrahedra in Rips complex when $L_2$ is considered, causing OOM error. As shown in the right panel of Figure \ref{fig:grid_cell}, when $\rho\leq0.075$, TRLS-H1 performs similarly to PHL-II and TopoGraph. And TopoGraph leads to smaller Wasserstein distance when $\rho>0.075$.

The right panel of Figure \ref{fig:grid_cell} shows that the total time of scoring/sampling landmarks and computing Rips PD on them is around two orders of magnitude lower than that of computing Rips PD on the full distance input. The result shows that TRLS-H1 adds a moderate scoring cost, but substantially reduces the subsequent Rips computation. This supports the use of TRLS as a practical preprocessing step when full Rips PD is computationally expensive.


\subsection{Robustness of TRLS}
We test the robustness of TRLS against outlier and noise in Figure \ref{fig:torus_in_circle_outlier_robustness}. The performance of random sampling, FPS and DA-FPS degrades with increasing outlier ratios. The rest method remain rather stable, with TRLS having the largest signal fraction and smallest Wasserstein distance, followed by TopoGraph and PHL-II. 

At the bottom row, TRLS has the highest signal fraction over different noise stds and similar Wasserstein distance to TopoGraph, follow by PHL-II. Overall, TRLS performs in a robust way against outlier and noise.

\begin{figure}[h]
    \centering
\includegraphics[width=.95\linewidth]{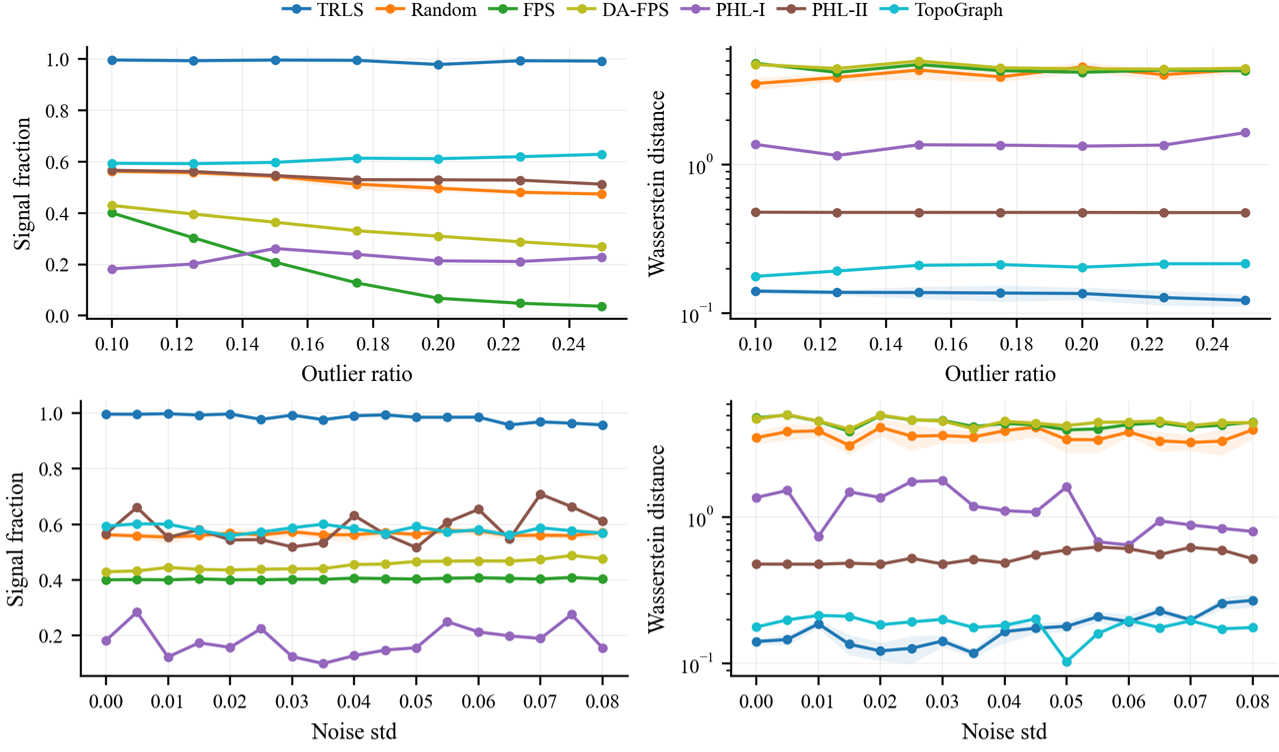}
    \caption{
Results under different outlier ratios (top row, no Gaussian noise) and noise stds of Gaussian perturbation (bottom row) on the \textsc{torus-in-circle} dataset when $\rho=0.3$.
    }
\label{fig:torus_in_circle_outlier_robustness}
\end{figure}


\section{Discussion}
\paragraph{Time \& Space Complexity.}
The computation bottleneck of TRLS lies in computing the kernel space of $L_q$. Given the eigen budget $R$ (user-specified number of topological features in interest or the upper bound on the number of near-zero eigenvectors), sparse Lanczos eigensolver \citep{lehoucq1998arpack,virtanen2020scipy} achieves $O(RM_q)$ time and $O(M_q+RN_q)$ space, with $M_q$ as the number of non-zeros in $L_q$ and $N_q$ as $q$-simplex count. Appendix \ref{appendix: scale-up} shows TRLS scales competitively with PHL and TopoGraph.

\paragraph{Hyperparameter Sensitivity Analysis.} 
We include a sensitivity study on $k \& \tau$ in Appendix \ref{app:p0a_sensitivity}. Across all datasets, the default values lie in a stable region rather than an isolated optimal point. Hence the default setting is a practical choice.

\paragraph{Ablation Study.} We conduct ablation study in Appendix \ref{appendix: ablation on RKDE}. Without RKDE, TRLS is less stable under increasing outliers and performs poorly. Poor performance of using reliability weights alone shows that RKDE improves complex construction, but cannot replace the topology-aware Hodge scoring. This is because RKDE alone does not encode topology and it is a robust complex construction prior rather than topology-aware score.

\paragraph{Advantage of TRLS over TopoGraph}
Beyond the frequent performance advantage of TRLS over TopoGraph, TRLS supports homology dimension-specific landmark selection through the corresponding Hodge Laplacian, as shown in Appendix \ref{appendix: 1}. For example, on \textsc{Sphere-in-circle}, TRLS-H1 concentrates on the circle, whereas TRLS-H2 selects the sphere. TopoGraph does not explicitly encode global homological structure or distinguish between homology dimensions. Therefore TopoGraph selects points from the sphere, circle, and connector without distinguishing their topological roles.


\section{Conclusion}

We proposed TRLS, a topology-aware landmark selection method for persistent homology. TRLS uses RKDE-guided complex construction and Hodge-Laplacian harmonic energy to score points by their contribution to global topological structures, rather than by geometry or local topology.

Across synthetic and real-world datasets, TRLS often outperforms the baselines based on geometry, graph coarsening or local PH. Its dimension-specific variants show that biological or physical prior knowledge can guide targeted landmark selection, while robustness and runtime experiments demonstrate its practical value for noisy and large point clouds.




\backmatter

\bmhead{Supplementary information}

Code and data can be found at supplementary material.

\begin{appendices}

\section{Additional Experiments}
\subsection{Persistence Diagrams of the Landmarks Selected by TRLS Using Different Homology Dimensions}
\label{appendix: 1}

Figure \ref{fig:ablation_PD_SIC} and \ref{fig:ablation_PD_TIC} shows the PDs on the landmarks selected by TRLS and its variants. Figure \ref{fig:ablation_PD_SIC_point} and \ref{fig:ablation_PD_TIC_point} show that the TRLS Hodge-derived scores are strongly dimension-aware and that the fused score preserves the topologically relevant points from both homological dimensions. Corresponding results of signal fraction and Wasserstein distance are given in Figure \ref{fig:ablation_dimension}.

\begin{figure}[h]
    \centering
\includegraphics[width=.8\linewidth]{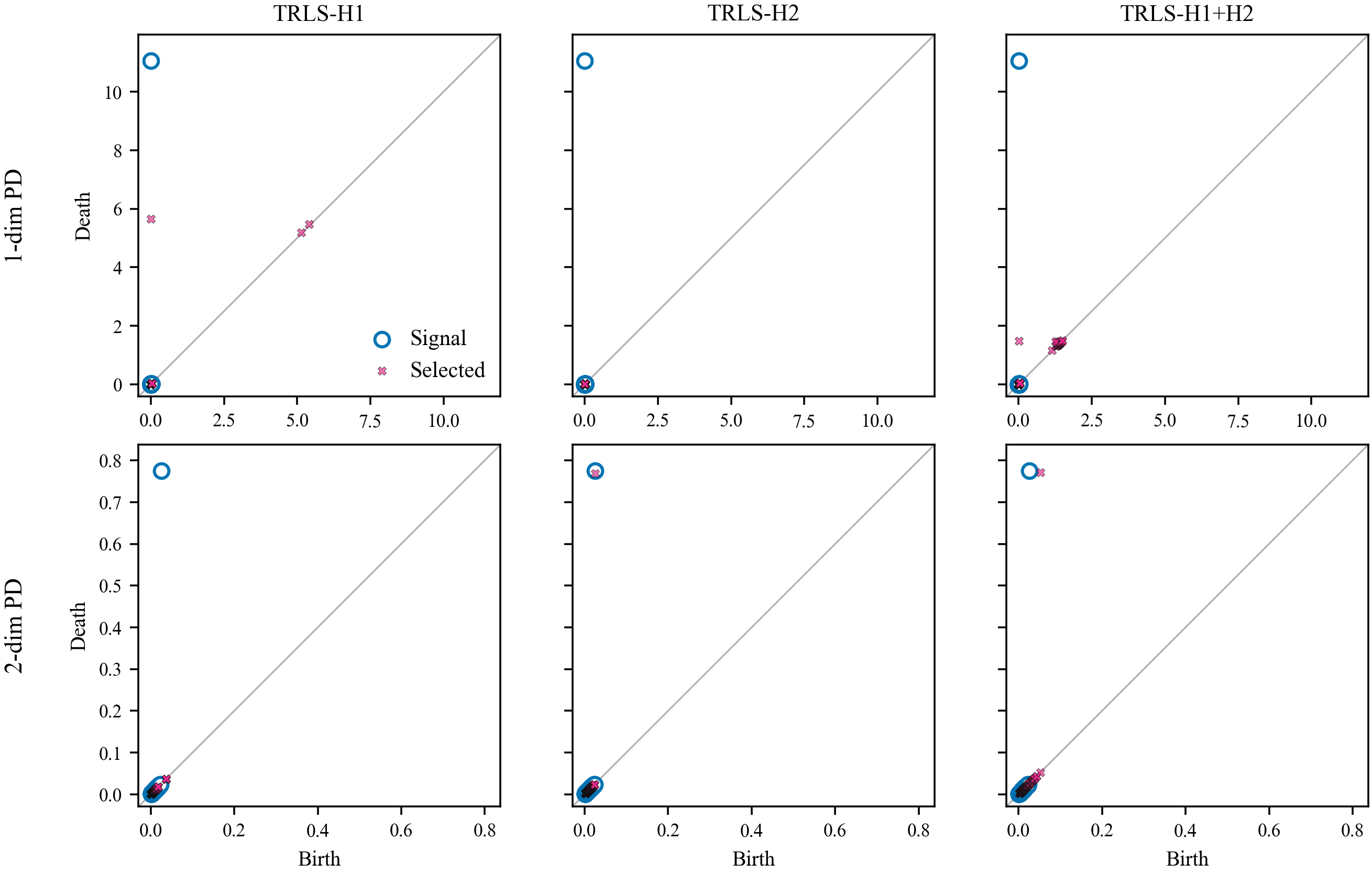}
\caption{
PDs on the landmarks selected by TRLS on \textsc{sphere-in-circle} dataset.
}
 \label{fig:ablation_PD_SIC}
\end{figure}

\begin{figure}[h]
    \centering
\includegraphics[width=.8\linewidth]{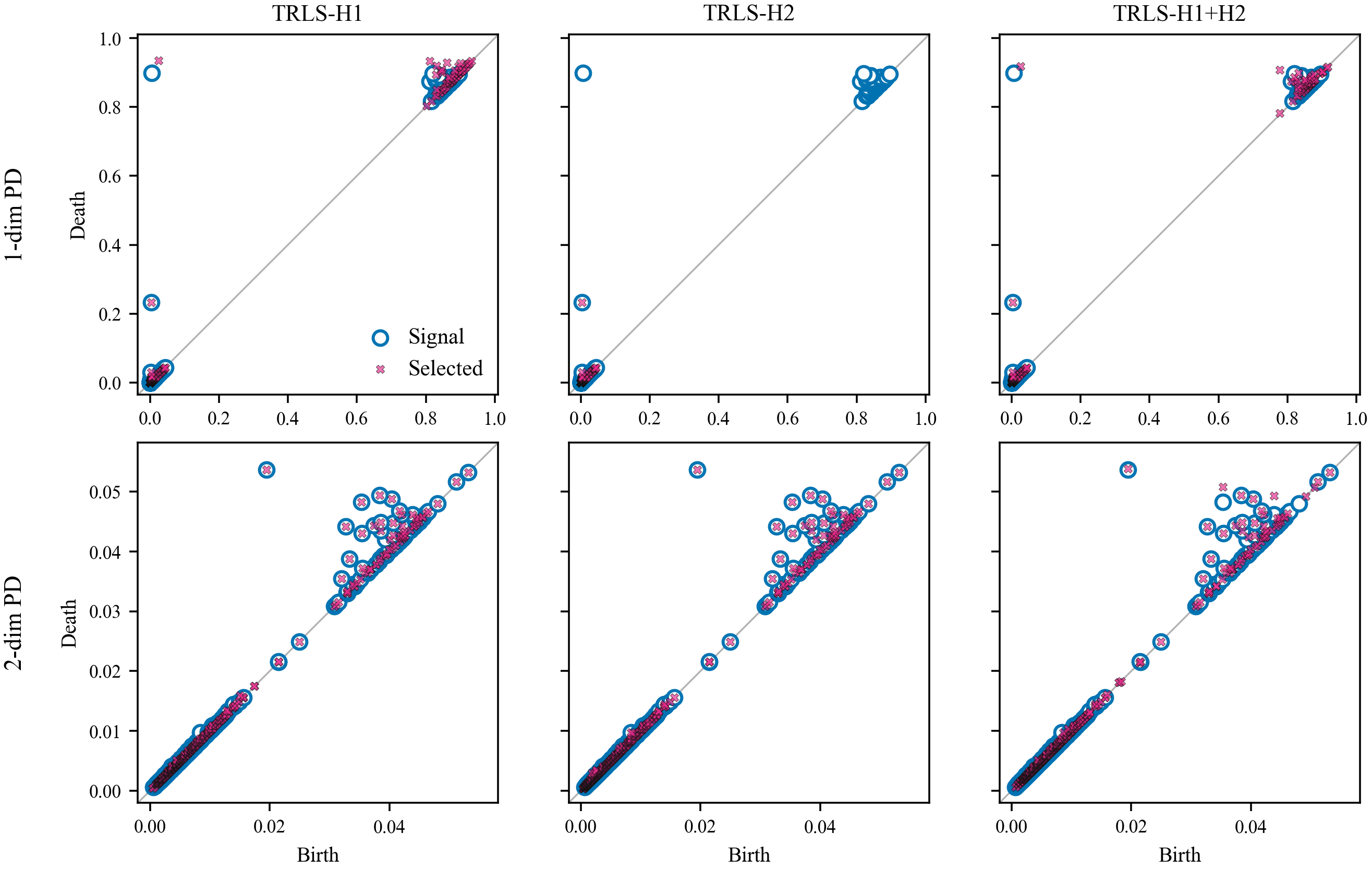}
\caption{
PDs on the landmarks selected by TRLS on \textsc{torus-in-circle} dataset.
}
\label{fig:ablation_PD_TIC}
\end{figure}

\begin{itemize}
    \item \textbf{Dimension-wise selectivity.}
    On \textsc{sphere-in-circle}, as shown in Figure \ref{fig:ablation_PD_SIC_point}, TRLS-H1 primarily selects the circular component, while the TRLS-H2 almost exclusively selects the sphere. At a large sampling density $\rho=0.4$, TRLS-H1 selects $81.1\%$ circle points (the rest selected outliers and points from line connector would shorten the H1 persistence, as shown in the first subplot in Figure \ref{fig:ablation_PD_SIC}), whereas TRLS-H2 selects $99.7\%$ sphere points. This shows that the Hodge-derived score can separates 1-dimensional and 2-dimensional topological evidence. On \textsc{torus-in-circle} dataset, as shown in Figure \ref{fig:ablation_PD_TIC}, TRLS-H1 alone can give a good recovery of the PDs because the torus has both 1-dimensional and 2-dimensional topological features. Results under different sampling densities are provided in Figure \ref{fig:ablation_dimension}.

    \item \textbf{Effect of score fusion.}
    As shown in Figure \ref{fig:ablation_PD_SIC_point}, the fused score balances both dimensions: on \textsc{sphere-in-circle}, it selects both sphere and circle points while suppressing connector and outlier points. On \textsc{torus-in-circle} in Figure \ref{fig:ablation_PD_TIC_point} and \ref{fig:ablation_dimension}, it selects $63.9\%$ torus points and $35.8\%$ ring points, with only $0.3\%$ connector points. Thus, fusion does not simply average noisy scores; it preserves multiple topologically relevant regions.
    \begin{figure}[h]
    \centering
\includegraphics[width=.98\linewidth]{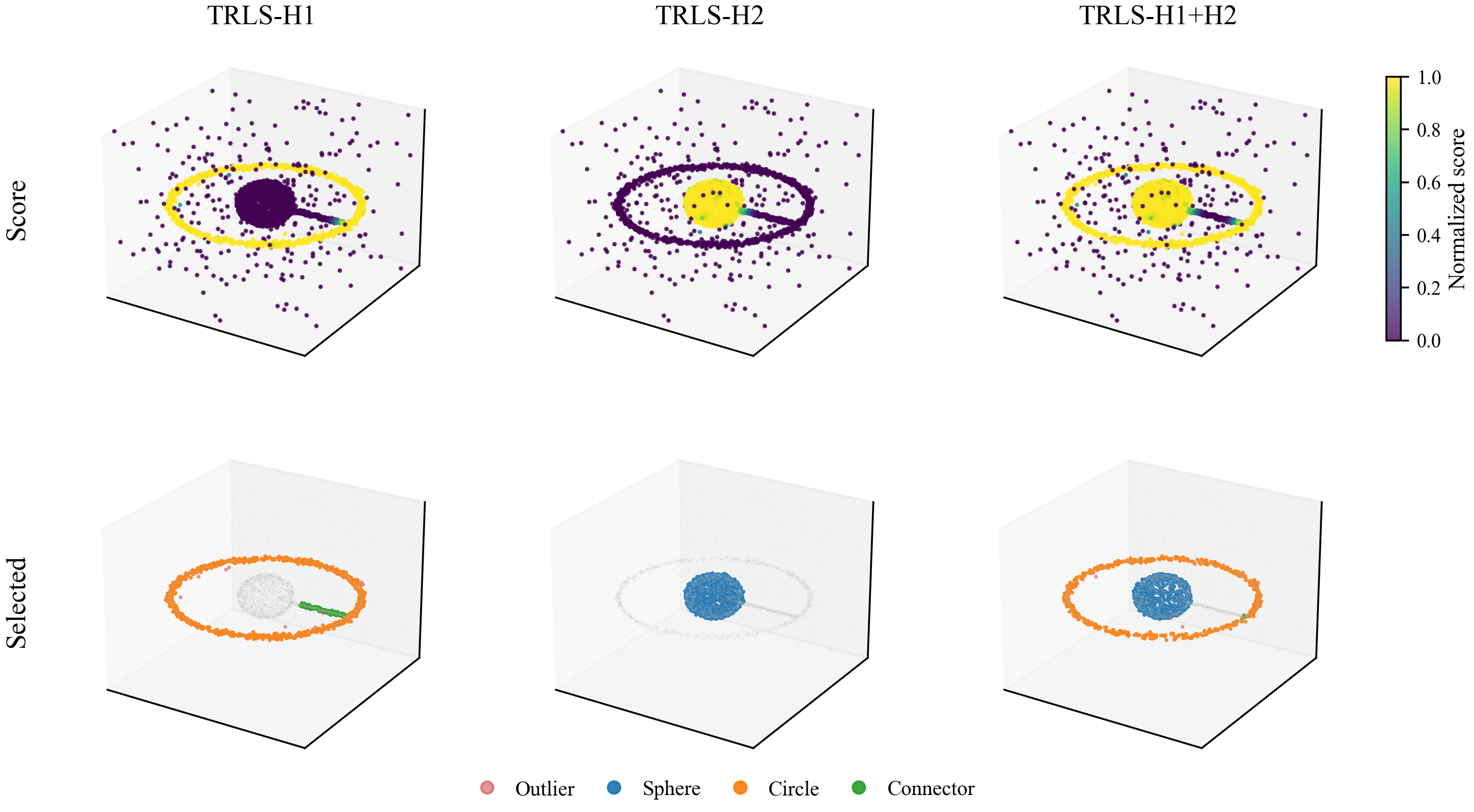}
\caption{
The top row shows Hodge scores of TRLS on sphere-in-circle dataset. The bottom row shows the seleced landmarks when $\rho=0.4$.
}
    \label{fig:ablation_PD_SIC_point}
\end{figure}
\begin{figure}[h]
    \centering
\includegraphics[width=.98\linewidth]{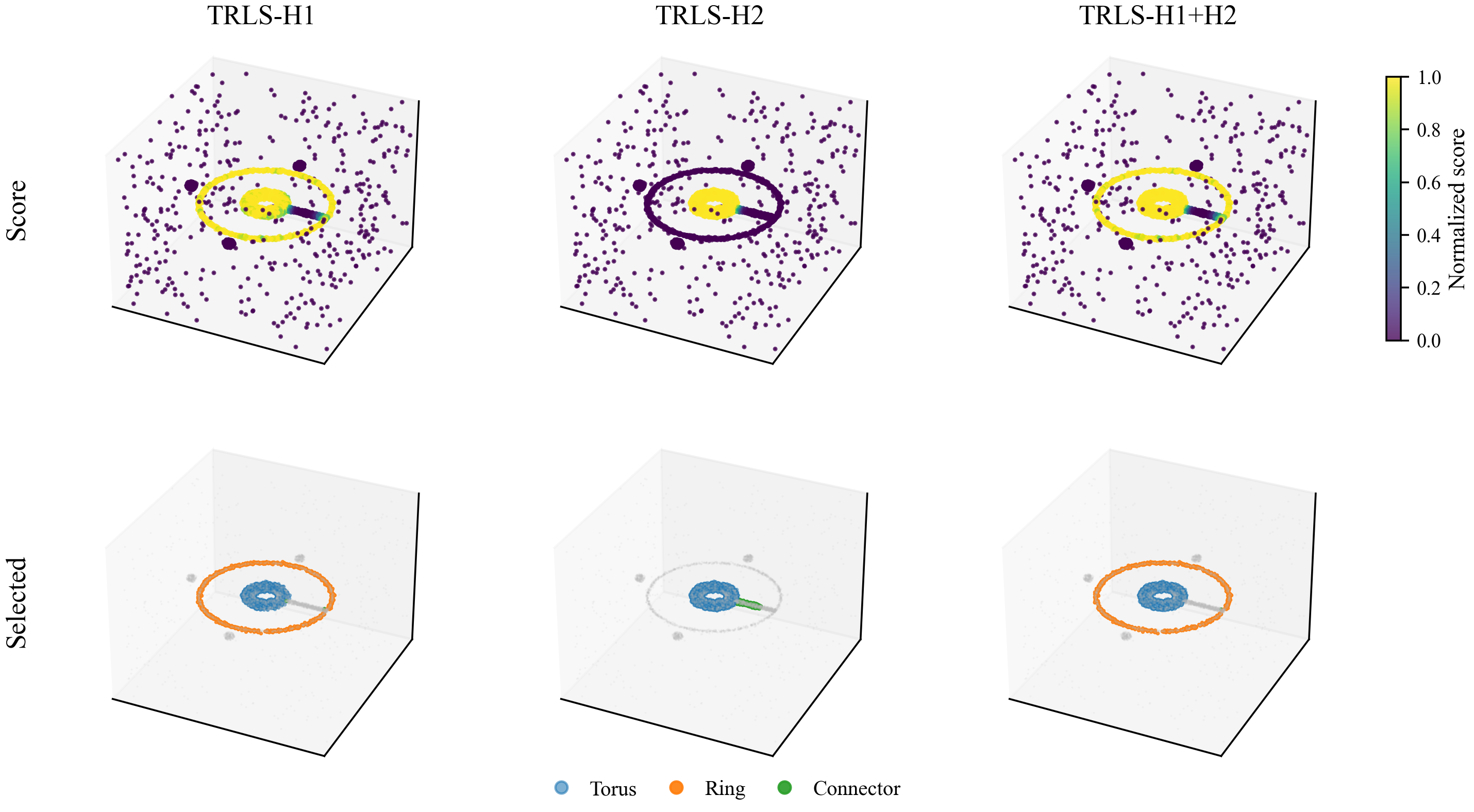}
\caption{
The top row Hodge scores of TRLS on torus-in-circle dataset. The bottom row shows the seleced landmarks when $\rho=0.4$.
}
    \label{fig:ablation_PD_TIC_point}
\end{figure}

\begin{figure*}[h]
    \centering
\includegraphics[width=.98\linewidth]{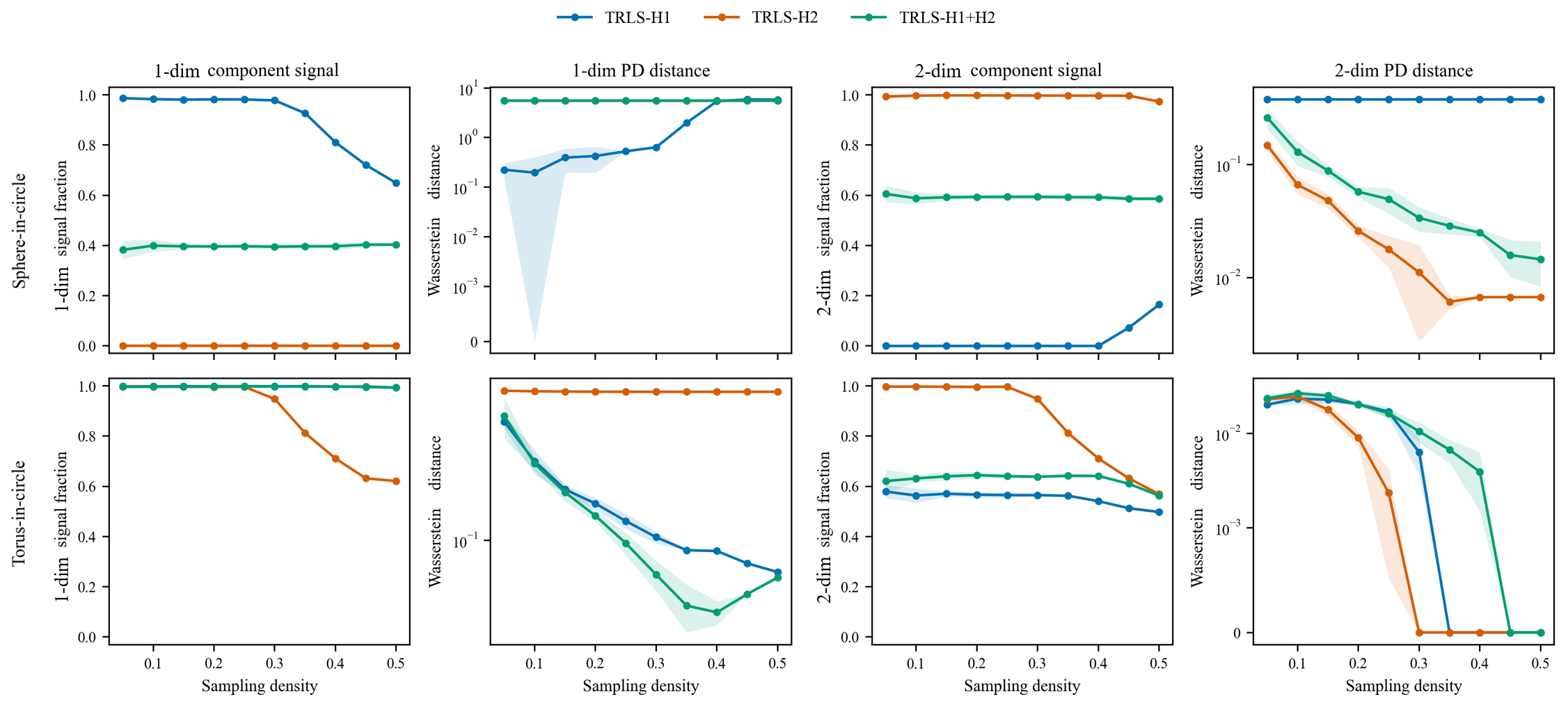}
\caption{
Evaluation of TRLS using only the H1 score, only the H2 score, and the fused H1+H2 score on \textsc{sphere-in-circle} and \textsc{torus-in-circle} datasets. The H1 component signal fraction means that we only consider the point cloud part that has non-empty 1-dimensional topological feature as topological relevant points. In \textsc{sphere-in-circle}, the topological relevant points $X_\mathrm{r}$ are the points in circle. In \textsc{torus-in-circle}, they are the circle and tours. Similarly, for H2 component signal, the topological relevant points are sphere for \textsc{sphere-in-circle} and torus for \textsc{torus-in-circle}.
}
    \label{fig:ablation_dimension}
\end{figure*}

    \item \textbf{Point cloud region having topological features with multiple homology dimensions in \textsc{torus-in-circle} dataset.} As shown in the second column of Figure  \ref{fig:ablation_dimension}, a counterintuitive phenomenon is that considering H2 in TRLS in addition to H1 can reduce the Wasserstein distance between 1-dimensional PDs on \textsc{torus-in-circle} dataset. At $\rho=0.4$, TRLS-H1 achieves $W=0.0896$, while the fused score further improves it to $W=0.0479$. This is because the torus part has both 1-dimensional and 2-dimensional topological features. The consideration of H2 will emphasize the torus part.

    \item \textbf{Overall implication.}
    The results support the intended role of TRLS: Hodge-derived scores can identify topologically relevant regions for each homology dimension, and the fused version provides a practical version when structures of all homology dimensions are considered to be preserved simultaneously.
\end{itemize}

\subsection{Signal Fraction Results on Scientific Domain Datasets}
\label{appendix: 2}

\begin{figure}[h]
    \centering
            \includegraphics[width=.8\linewidth]{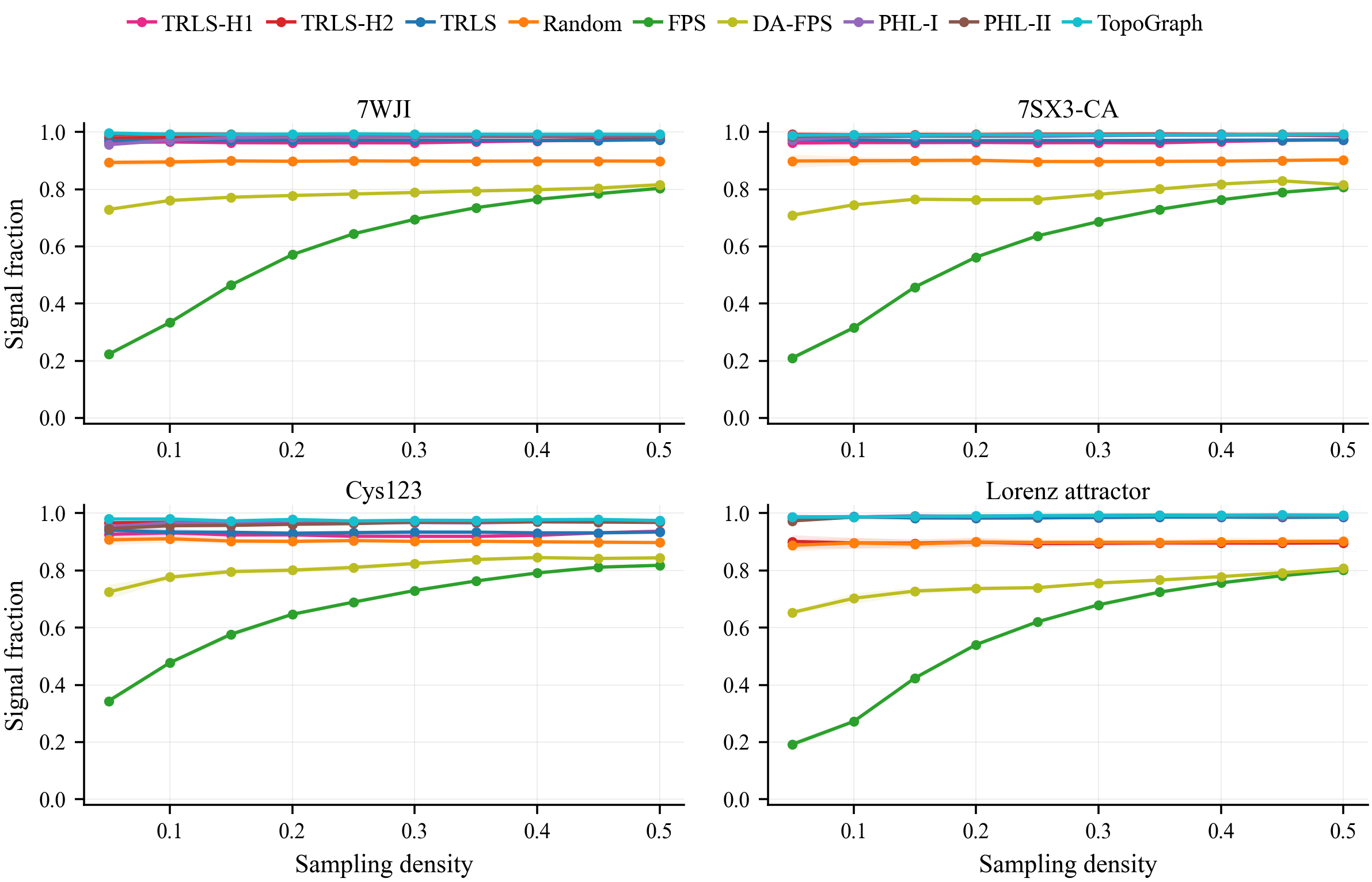}
    \caption{
    Signal fraction of selected landmarks on four scientific domain datasets.
    }
    \label{fig:nalcn_signal_fraction}
\end{figure}

We compare TRLS and its variants with Random, FPS, PHL and TopoGraph across different sampling densities in Figure \ref{fig:nalcn_signal_fraction}. On the NALCN channelosome 7WJI, 7SX3-CA and Lorenz datasets, TRLS and its variants have similar signal fraction (close to 1.0) to TopoGraph, PHL and its variants. This shows that both TRLS, PHL and TopoGraph can identify outliers. FPS performs the worst since it only considers spatial coverage and hence tends to choose outliers. DA-FPS improves over FPS, but it still suffers from the same issue. Random sampling is better that FPS and DA-FPS. It selects less outliers. But random sampling is worse than PHL, TRLS and their variants since random sampling does not consider geometry or topology of a point cloud at all. 

\subsection{Sensitivity of $k\&\tau$ to Scale Selection }
\label{app:p0a_sensitivity}

We evaluate whether the default TRLS hyperparameters are stable or merely a lucky configuration. 
This experiment focuses on the two parameters used by the RKDE-guided scale heuristic: the neighborhood size $k_\alpha$ and the weighted quantile $\tau_\alpha$. 
For each dataset, we sweep
$k_\alpha \in \{2,4,6,8,10,12,16,20,24,32\}, 
\tau_\alpha \in \{0.50,0.60,0.70,0.75,0.80,0.85,0.90,0.95,0.975,0.99\},
$
with the default setting marked by a cyan circle in Figure~\ref{fig:p0a_sensitivity}. 
We report both the Wasserstein distance between the selected PD and the signal PD, and the fraction of selected landmarks that belong to the signal set.

\begin{figure}[h]
    \centering
    \includegraphics[width=\textwidth]{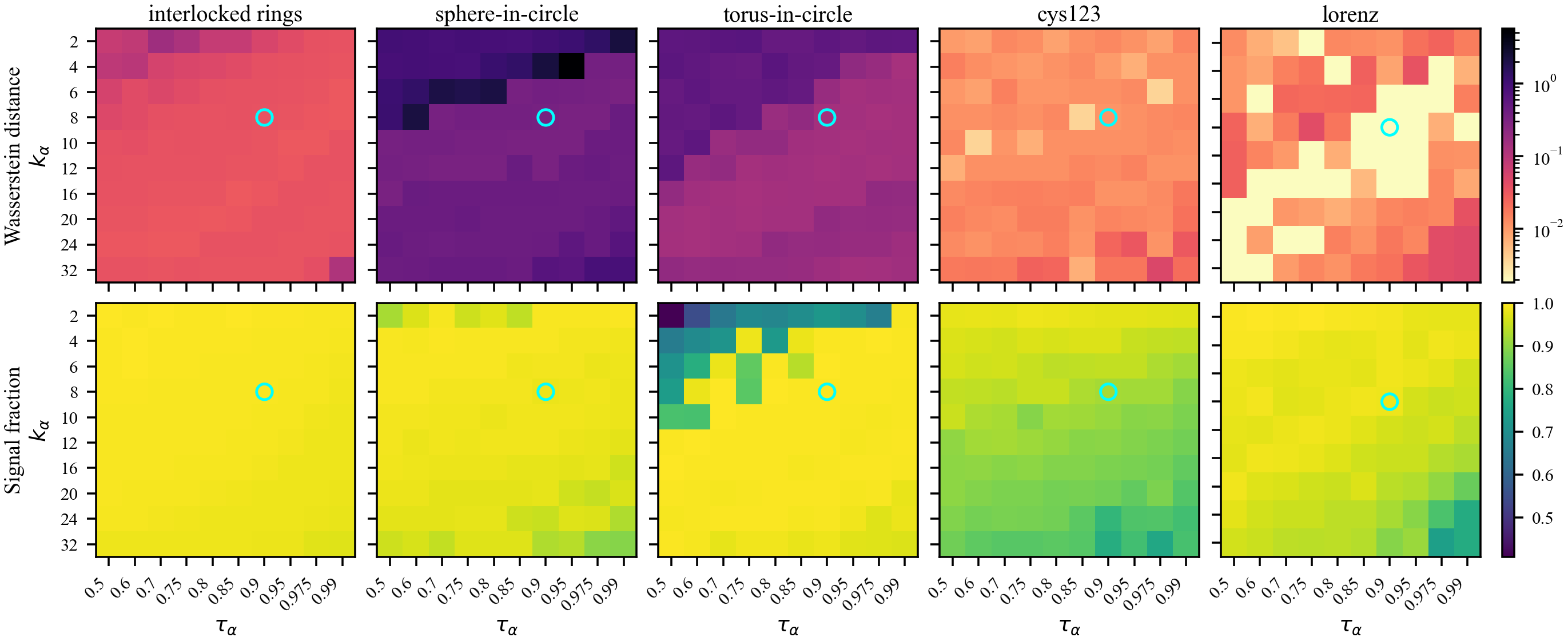}
    \caption{
    Sensitivity of TRLS to the scale heuristic parameters $k_\alpha$ and $\tau_\alpha$ on non-NALCN datasets. 
    Each column corresponds to one dataset. The top row shows Wasserstein distance, and the bottom row shows signal fraction. 
    Cyan circles mark the default setting $k_\alpha=8,\tau_\alpha=0.90$. 
    The default setting lies in a stable region for most datasets, rather than being an isolated optimum.
    }
    \label{fig:p0a_sensitivity}
\end{figure}

\begin{figure}[h]
    \centering
    \includegraphics[width=.8\linewidth]{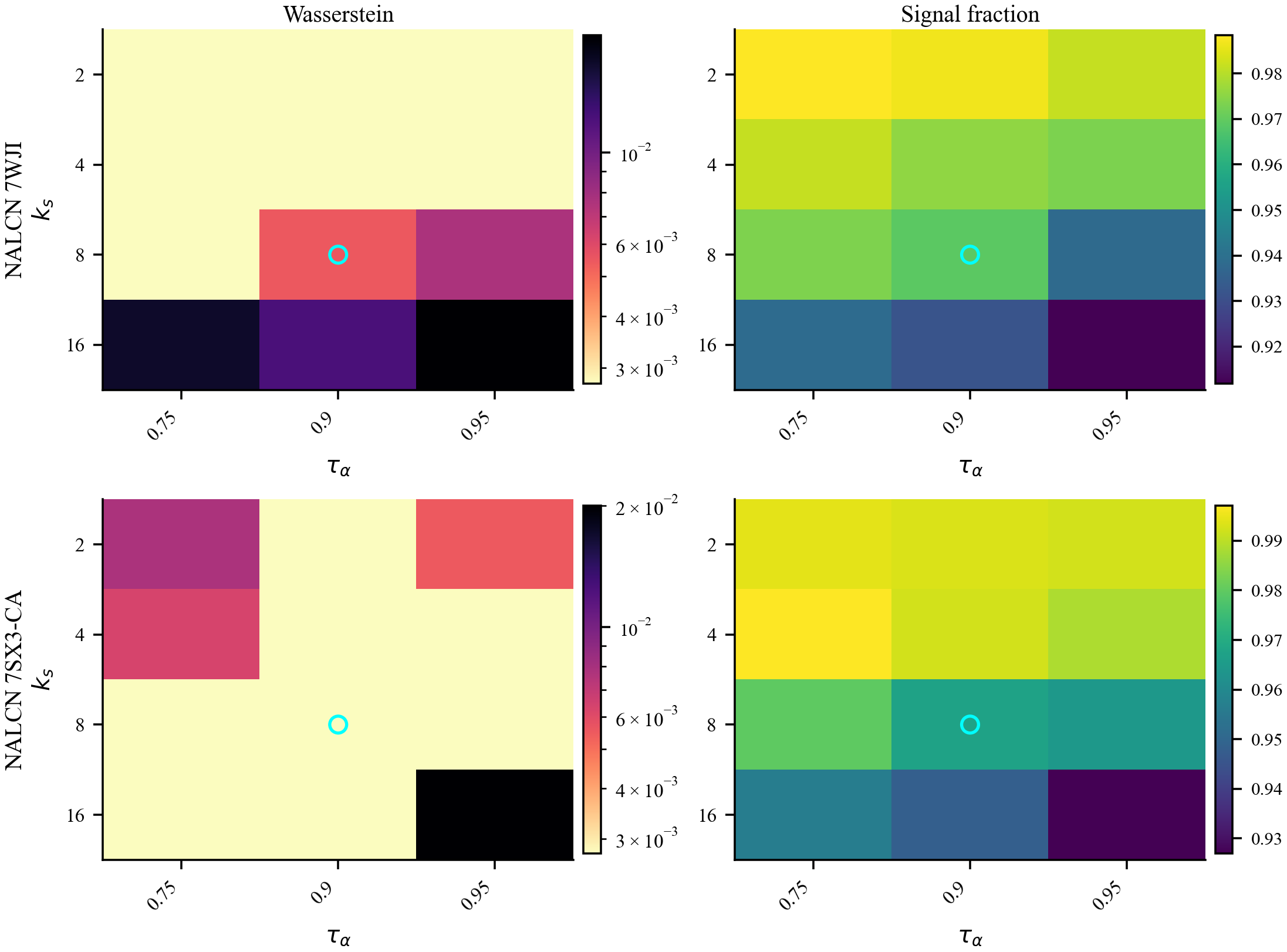}
    \caption{
    Sensitivity of TRLS to $k_\alpha$ and $\tau_\alpha$ on the two NALCN point clouds.
    }
    \label{fig:p0a_nalcn_sensitivity}
\end{figure}

The results show that the default setting is not an isolated lucky point. 
On \textsc{torus-in-circle} dataset, the default achieves Wasserstein distance $0.1427$ and signal fraction $0.9964$, while the best grid point only slightly improves the distance to $0.1342$. 
On \textsc{sphere-in-circle} dataset, the default obtains $0.3517$ with signal fraction $0.9833$, close to the best observed distance $0.3299$. 
On \textsc{interlocked rings}, the default Wasserstein distance is $0.0357$, close to the median behavior across the grid. 
On Cys123, the default is also competitive, although several lower or middle scale settings perform better, indicating that \textbf{the default setting is robust but not always optimal}. 
Lorenz is more sensitive in Wasserstein distance. The default keeps a high signal fraction of $0.9850$.


We further evaluate the scale sensitivity on the two NALCN point clouds in Figure \ref{fig:p0a_nalcn_sensitivity}. 
For these two datasets, we restrict the search to 
$k_\alpha \in \{2,4,8,16\}$ and 
$\tau_\alpha \in \{0.75,0.90,0.95\}$, since larger scales are substantially more expensive on protein point clouds. 
The results show that the default setting $k_\alpha=8,\tau_\alpha=0.90$ remains in a usable region for both datasets. 
On 7WJI, it gives Wasserstein distance $0.00548$ and signal fraction $0.9691$; on 7SX3-CA, it gives signal fraction $0.9673$. 
Several smaller scale settings obtain zero Wasserstein distance and slightly higher signal fraction, which suggests that NALCN is less demanding in terms of scale for the persistence features used in evaluation. 
At the same time, increasing $k_\alpha$ to $16$ consistently lowers the signal fraction, indicating that overly large scales can mix signal atoms with outliers and reduce selection purity. 

\textbf{These results support the intended role of the scale heuristic, which provides a stable operating region without requiring per-dataset tuning.} 

\subsection{Sensitivity to Simplex Seed Selection}
\label{app:simplex-seed-sensitivity}

Figure~\ref{fig:p0a-simplex-quantile} studies the sensitivity of TRLS to the simplex seed quantile $\tau_\sigma$. This parameter is used after the Hodge energy has been computed: TRLS selects high energy simplices as seeds, lifts their support to a point score, and then selects landmarks from the resulting score. \emph{This hyperparameter $\tau_\alpha$ also controls the alpha complex construction.}

For each $\tau_\sigma$, we report the Wasserstein distance to the clean signal PD reference, thesignal fraction, and the energy leakage ratio. Let $v_i(\tau_\sigma)$ be the final normalized Hodge-derived point score. 

The energy leakage ratio is
\[
    \frac{v_{\mathrm{irr}}(\tau_\sigma)}
         {v_{\mathrm{rel}}(\tau_\sigma)}
    \quad\text{or}\quad
    \frac{v_{\mathrm{out}}(\tau_\sigma)}
         {v_{\mathrm{rel}}(\tau_\sigma)},
\]
where $\mathrm{rel}$ denotes the clean topological signal, $\mathrm{irr}$ denotes non topological nuisance components, and $\mathrm{out}$ denotes injected outliers. Thus, lower leakage means that the final TRLS Hodge-derived  score assigns less mass to nuisance or outlier points relative to the relevant topological region.

As shown in Figure \ref{fig:p0a-simplex-quantile}, the default value $\tau_\sigma=0.90$ lies in a stable region across the seven datasets. Very small $\tau_\sigma$ selects too many seed simplices and can spread the lifted score into nuisance regions. Very large $\tau_\sigma$ can make the seed set too narrow. In \textsc{sphere-in-circle} and \textsc{torus-in-circle}, this keeps the signal fraction high but worsens the persistence distance. The protein and Lorenz attractor datasets have several zero Wasserstein values, so these panels should be interpreted together with signal fraction and leakage. Overall, the result supports using a high but not extreme simplex seed quantile as the default.

\clearpage
\begin{figure}[h]
    \centering
    \includegraphics[width=.9\linewidth]{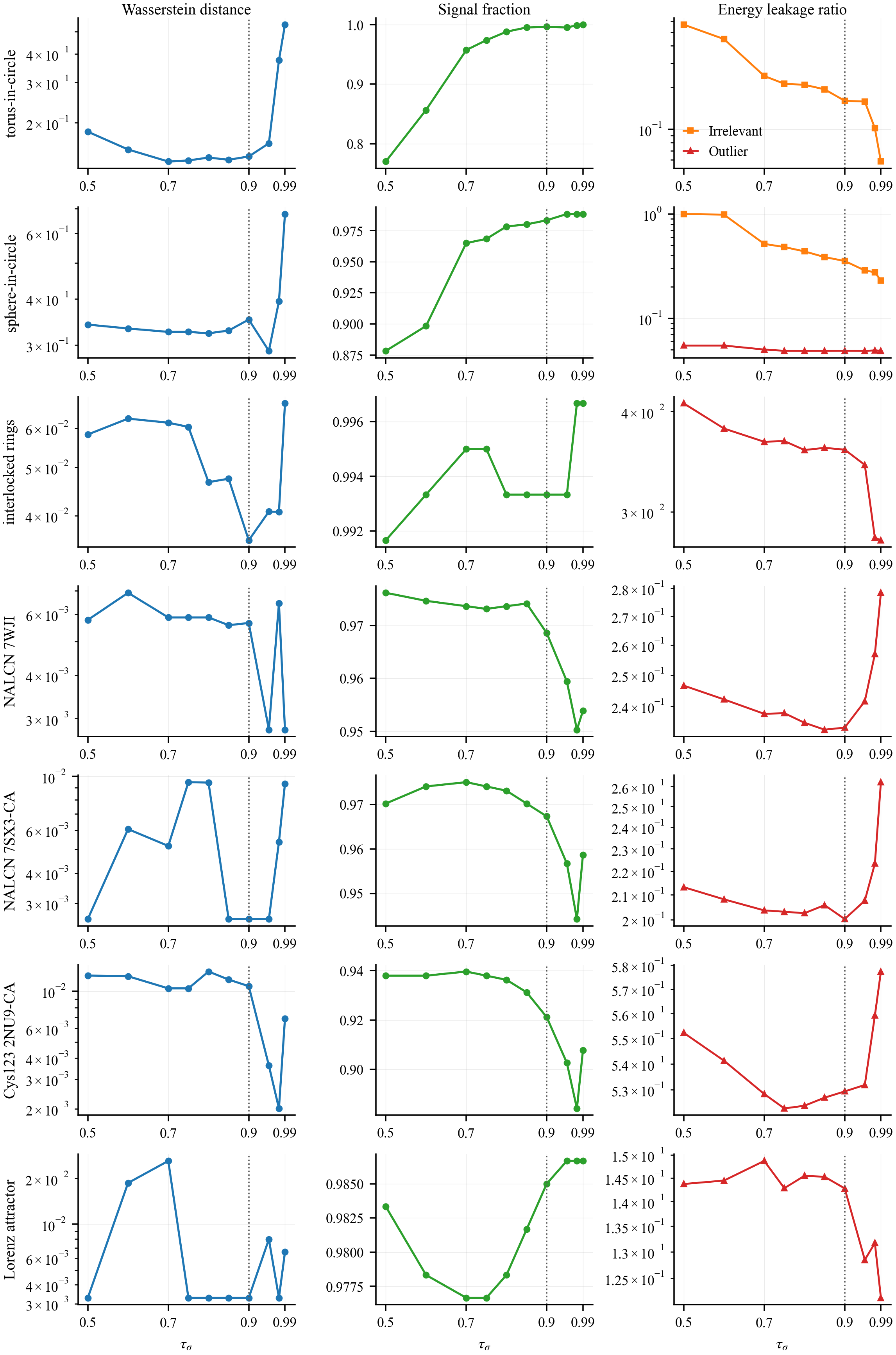}
    \caption{
    Sensitivity of TRLS to the simplex seed quantile $\tau_\sigma$ on the seven datasets. 
    The vertical dotted line marks the default value $\tau_\sigma=0.90$.
    The first column reports Wasserstein distance to the clean signal PD reference, the second column reports signal fraction, and the third column reports the energy leakage ratio from topologically irrelevant or outlier points to the topologically relevant signal.
    }
    \label{fig:p0a-simplex-quantile}
\end{figure}
\clearpage



\subsection{Choice of the Near Zero Eigenvalue Threshold}
\label{appendix: 5}

TRLS uses the near zero eigenspace of the Hodge Laplacian to estimate the harmonic subspace. In practice, we set the numerical threshold to $\eta=10^{-8}$. This choice is supported by the eigenspectrum in Figure~\ref{fig:nalcn_eigenspectrum}. 

For both NALCN 7WJI and NALCN 7SX3-CA, and for both $q=1$ and $2$, the spectra show a clear separation between numerical zero eigenvalues around machine precision and the first positive eigenvalues, which are several orders of magnitude larger. Therefore, $10^{-8}$ lies safely inside the spectral gap and separates the harmonic subspace from the non-harmonic components. This indicates that the threshold is not a sensitive tuned parameter in these experiments, but a numerical cutoff for distinguishing near zero eigenvalues from stable positive modes.

\begin{figure}[h]
    \centering
\includegraphics[width=\linewidth]{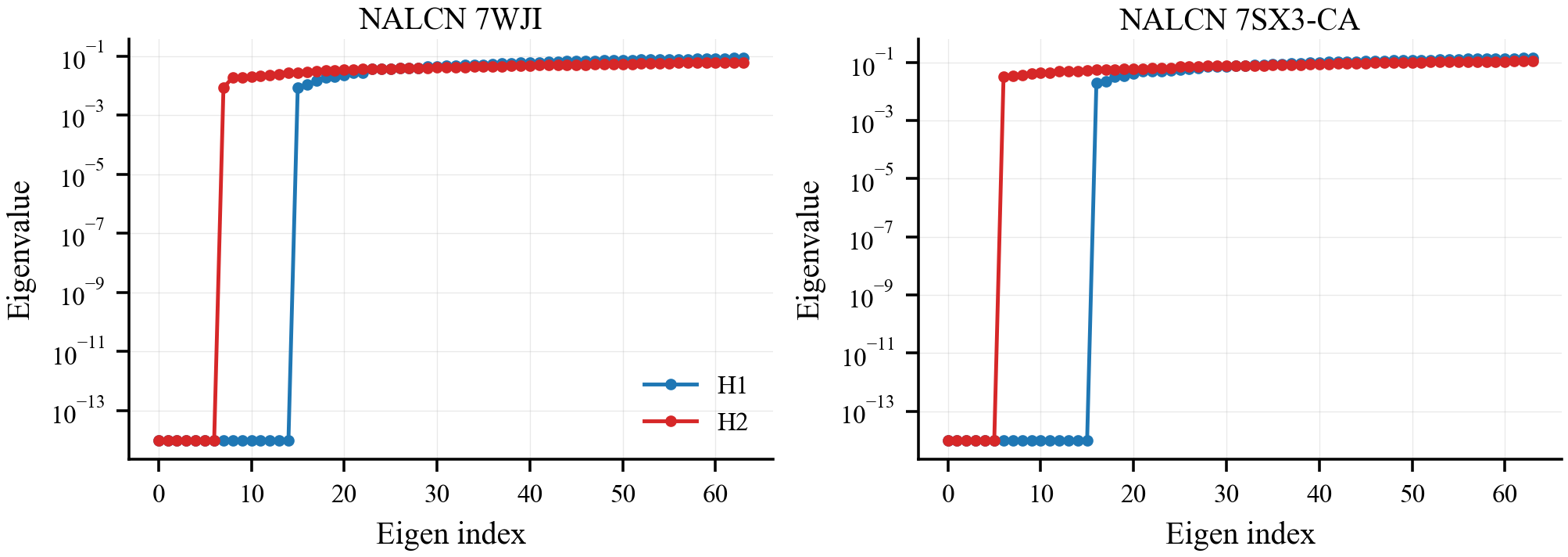}
    \caption{
    Small eigenvalue spectra of the default TRLS Hodge Laplacians on the two NALCN point clouds. 
    The near zero eigenvalues are separated from the positive spectrum by a clear gap, supporting the use of $10^{-8}$ as the numerical threshold.
    }
    \label{fig:nalcn_eigenspectrum}
\end{figure}

\subsection{Ablation on RKDE}
\label{appendix: ablation on RKDE}
\begin{figure}[h]
    \centering
    \includegraphics[width=\linewidth]{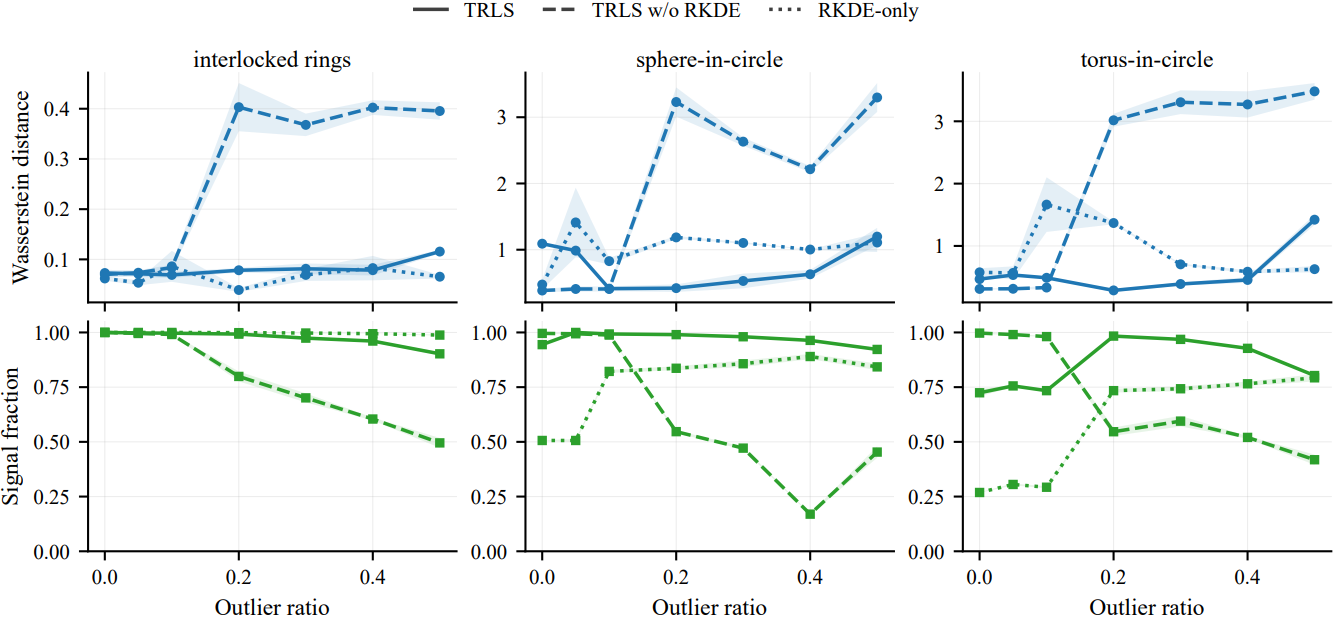}
    \caption{
    RKDE ablation on three synthetic datasets. 
    }
    \label{fig:rkde_ablation_tradeoff}
\end{figure}

Across the three synthetic datasets in Figure \ref{fig:rkde_ablation_tradeoff}, TRLS with RKDE achieves a more favorable balance between topological accuracy and signal preservation. As the outlier ratio increases, the variant without RKDE shows a clear rise in Wasserstein distance and a drop in signal fraction, especially on \textsc{sphere-in-circle} and \textsc{torus-in-circle}. In contrast, TRLS remains much more stable, indicating that RKDE reliability helps suppress outlier induced structures before Hodge based scoring. This confirms that RKDE is not only a density correction, but also a key component for robust landmark selection under increasing contamination.

On \textsc{interlocked rings}, which does not have topologically irrelevant region, RKDE-only produces similar results to TRLS. On \textsc{sphere-in-circle} and \textsc{torus-in-circle}, using RKDE alone generally leads to higher Wasserstein distance between PDs because it ignores the topology-related part (Hodge Laplacian) and only considers density. This demonstrates the necessity of Hodge Laplacian in TRLS. 

\subsection{Consistency Between Scoring and Evaluation Filtrations}
\label{appendix: 7}

We further examine whether the advantage of TRLS depends on using the same filtration for both scoring and evaluation in Figure \ref{fig:alpha-rips-consistency}. On the \textsc{torus-in-circle} dataset, TRLS computes its scores from the RKDE-guided alpha complex, while the selected landmarks are evaluated using either alpha persistence or Rips persistence. Both alpha and Rips evaluation gives the same best method, i.e., TRLS. This indicates that the main conclusion is not merely an artifact of matching the scoring and evaluation complexes.

\begin{figure}[h]
    \centering
\includegraphics[width=\linewidth]{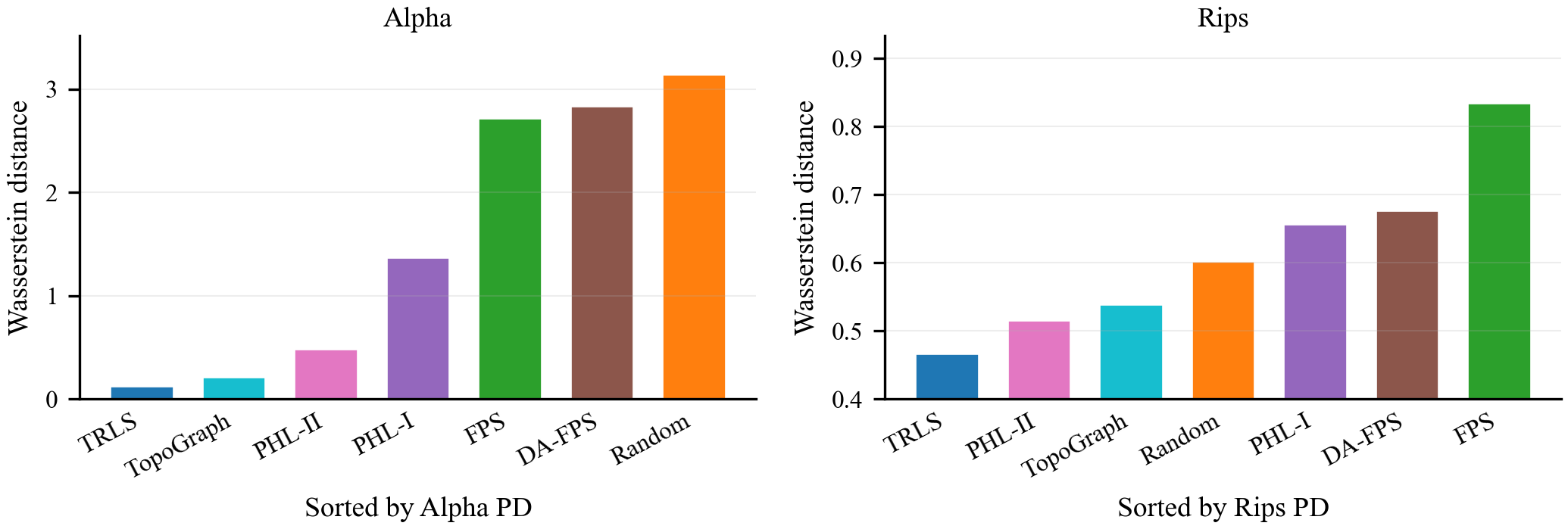}
    \caption{Consistency between alpha and Rips evaluation on the \textsc{torus-in-circle} dataset. The left panel reports Wasserstein distances sorted by alpha persistence evaluation, while the right panel reports Wasserstein distances sorted by Rips persistence evaluation.}
    \label{fig:alpha-rips-consistency}
\end{figure}


This result should not be interpreted as full filtration invariance. Rather, it provides an empirical sanity check that, for the strongest methods on this benchmark, topological scoring remains effective even when the final evaluation is performed with a different filtration.





\subsection{Stability Analysis}
\label{appendix: stability}

\begin{figure}[h]
    \centering
\includegraphics[width=.99\linewidth]{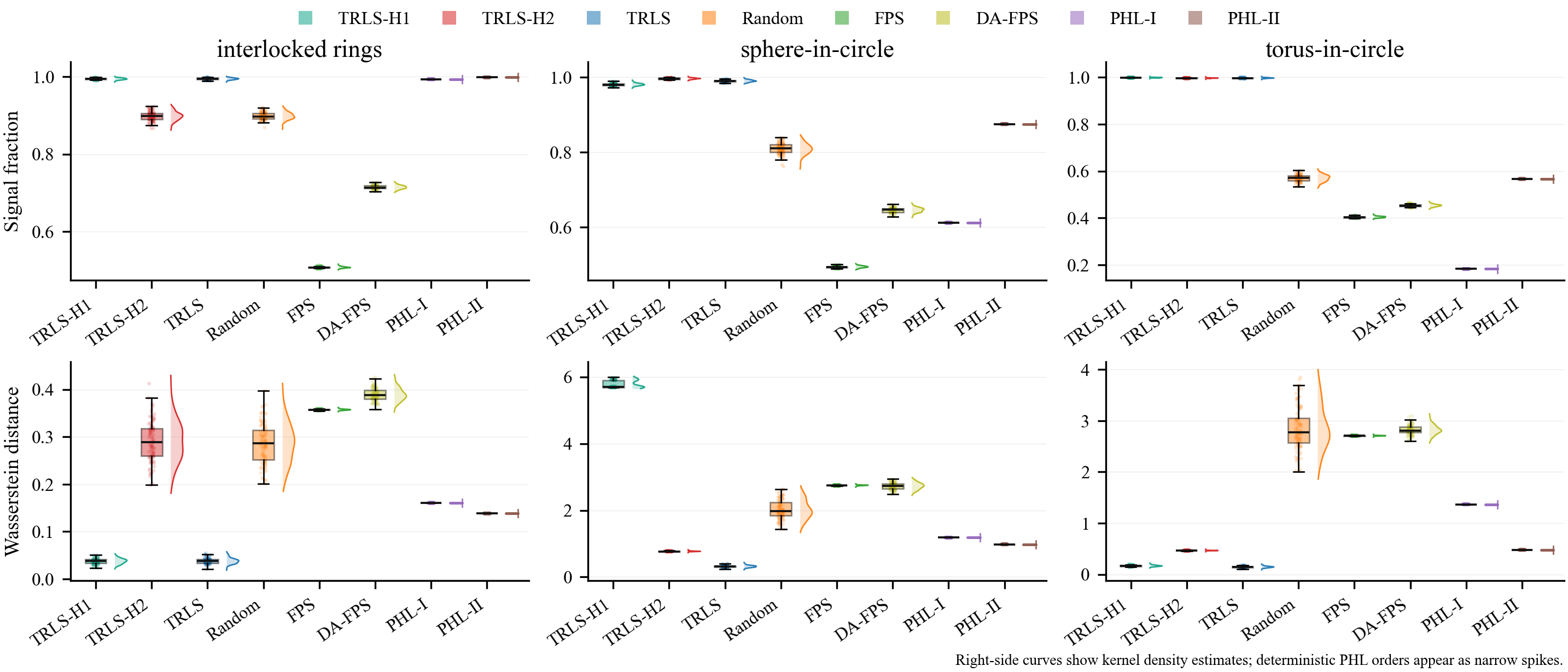}
    \caption{Sampling stability on synthetic data sets at $\rho=0.20$. TRLS keeps low Wasserstein distance and high signal fraction across repeated runs. The right side curve beside each box shows the empirical density estimate.}
    \label{fig:p0c_rho020_synthetic}
\end{figure}

\begin{figure}[h]
    \centering
\includegraphics[width=.99\linewidth]{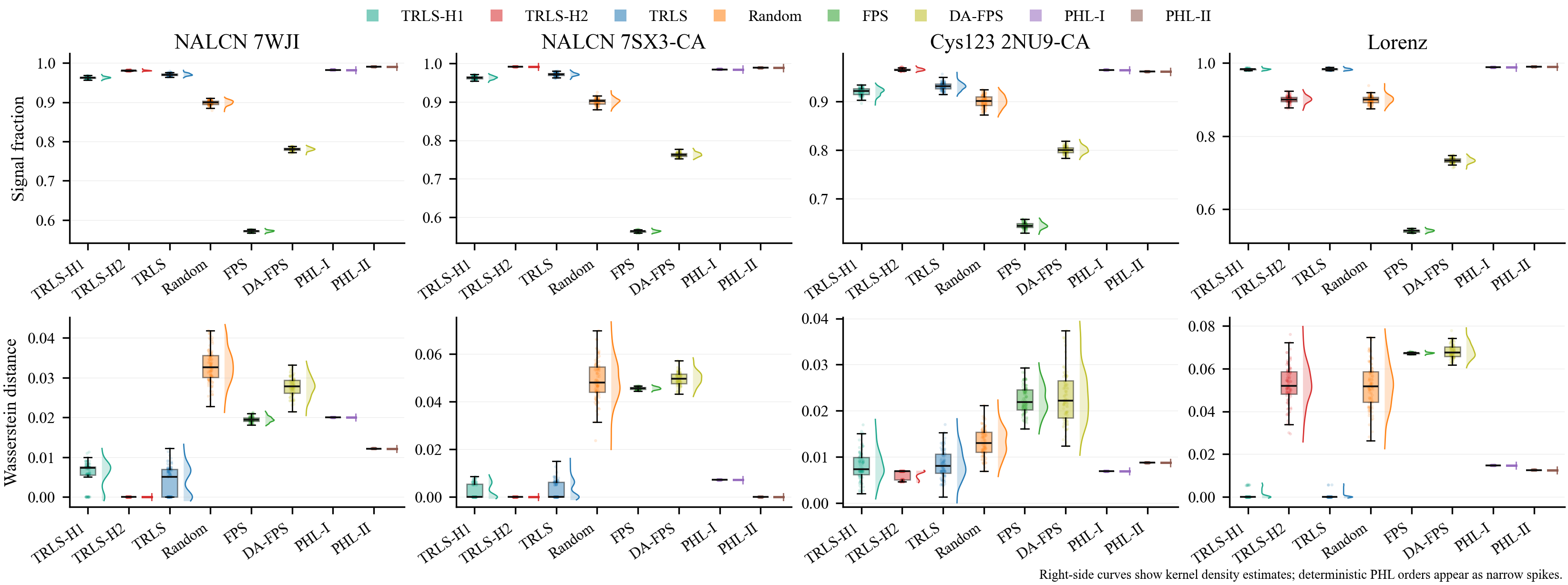}
    \caption{Sampling stability on real data sets at $\rho=0.20$. The results show that TRLS remains stable under repeated random tie breaking. The right side curve beside each box shows the empirical density estimate.}
    \label{fig:p0c_rho020_real}
\end{figure}



To further test whether the stability of TRLS depends on a single sampling density, we repeat the same stability experiment at $\rho=0.2$. For each stochastic method, we run 100 repetitions and report the distribution of signal fraction and Wasserstein distance in Figure \ref{fig:p0c_rho020_synthetic} and \ref{fig:p0c_rho020_real}. 

The results are consistent with the results in the main body of the paper. On the synthetic data sets, fused TRLS gives the most stable behavior on mixed topology data, while TRLS-H1 is only strong on data dominated by circles. On \textsc{sphere-in-circle}, TRLS-H1 fails because it ignores the sphere, which supports the need for dimension fusion. On the real-world datasets, TRLS variants keep low Wasserstein distance and high signal fraction. TRLS-H2 performs best on Cys123, matching the void oriented structure of this data set.

Overall, the results show that TRLS and its variants produce competitive and stable results.






\subsection{Scale-up test on the grid cell dataset.}
\label{appendix: scale-up}
We further evaluate the computational behavior on the grid cell dataset with fixed eigen budget $R=8$. For each method, we report the time for scoring and landmark selection at sampling density $\rho=0.1$. TRLS-H1 has similar time cost as that of TopoGraph and PHL methods. TRLS-H1, PHL, TopoGraph and DA-FPS are much slower than random sampling and FPS. However, the huge time cost of DA-FPS does not pay off. As shown in the main body of our paper, DA-FPS still selects outliers and topologically irrelevant points.

\begin{figure}[h]
    \centering
\includegraphics[width=0.6\linewidth]{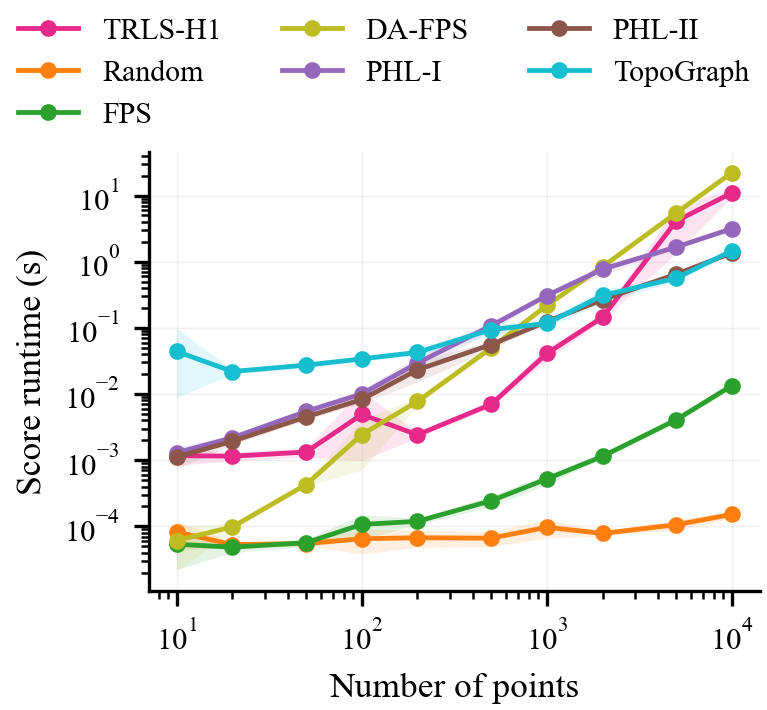}
    \caption{Scale-up runtime comparison on the grid cell dataset at sampling density $\rho=0.1$.}
    \label{fig:grid-cell-scale-up}
\end{figure}

\subsection{Scale validation by signal barcodes.}
\label{appendix: scale validation}
We verify whether the scale used by TRLS to build the alpha complex lies in the barcode intervals of the topologically relevant signal points. For each synthetic dataset, we compute the alpha barcode on the clean signal part and compare the TRLS scale with the birth and death times of the target features. Since GUDHI reports alpha filtration values as squared radii, we convert all intervals to radius scale before comparison.

\begin{figure}[t]
    \centering
    \includegraphics[width=0.75\textwidth]{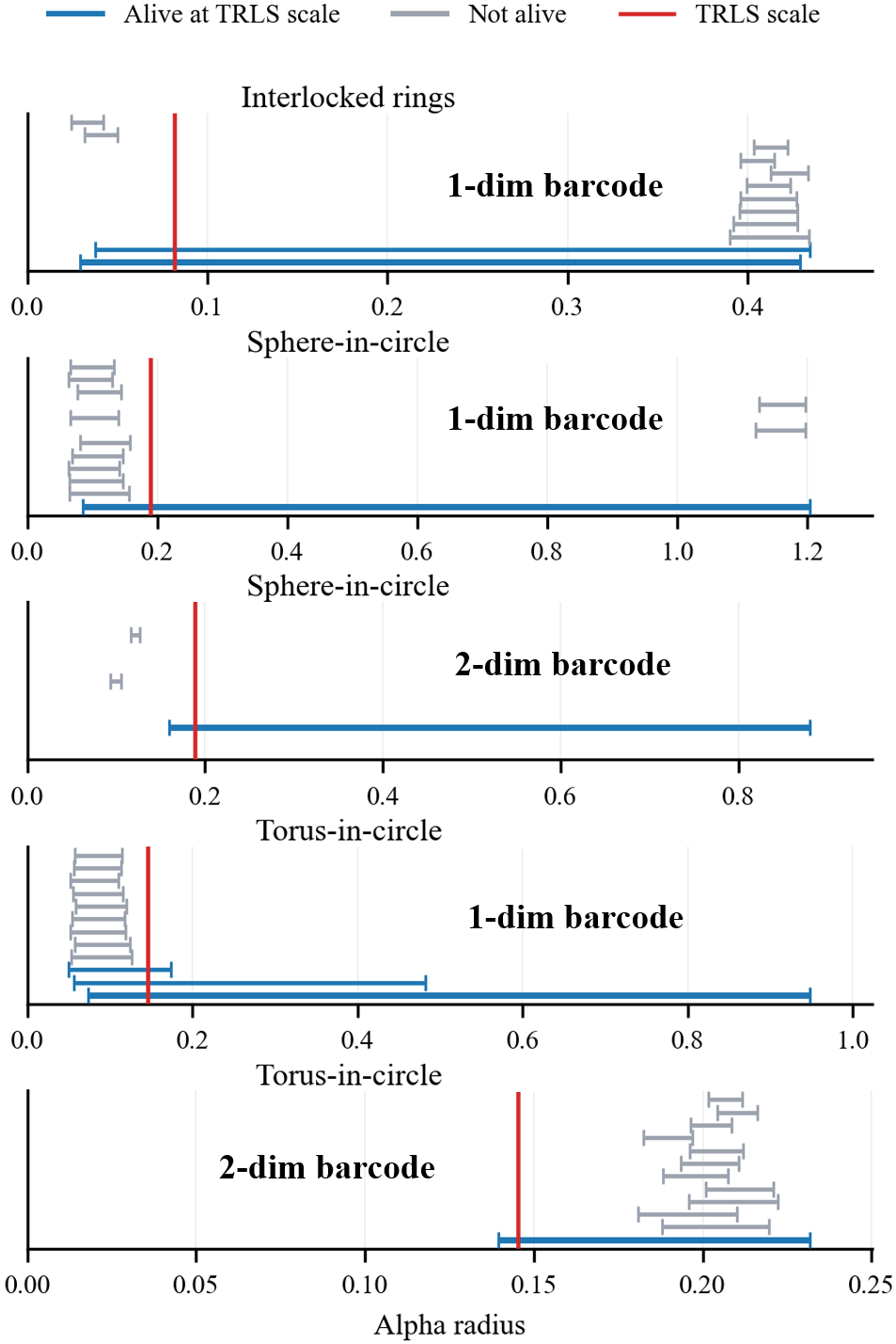}
    \caption{Validation of the TRLS scale on synthetic datasets.}
    \label{fig:scale-barcode-validation}
\end{figure}

As shown in Fig.~\ref{fig:scale-barcode-validation}, the TRLS scale falls inside the barcode interval of the target signal feature in all tested cases. This suggests that the selected scale is not an arbitrary local distance threshold, but is aligned with the persistence range where the relevant topology is alive.

\section{Experimental settings}
\label{appendix: exp setting}




All methods used in the experiments were run on a workstation with i7-12700K CPU, 128 GB RAM, NVIDIA GeForce RTS 3060 GPU and Windows 10 system.

\paragraph{Synthetic datasets.}
We use three synthetic point clouds to test whether landmark selection preserves the topologically relevant part of the data under outliers and irrelevant points. The \textsc{interlocked rings} dataset contains two linked noisy circles in $\mathbb{R}^3$ with uniformly sampled outliers. The two rings are signal and all outliers are non signal. The \textsc{sphere-in-circle} dataset contains a sphere and a larger circle sharing the same center, connected by a short line segment, with outliers. The sphere and circle are signal, while the connector is non signal. The \textit{torus-in-circle} dataset contains a torus inside a larger circle, connected by a line segment, with three solid spheres and outliers. The torus and circle are signal, while the connector, solid spheres are non signal. We evaluate 1-dimensional PD for \textsc{interlocked rings}, and both 1 and 2 dimensional PDs for \textsc{sphere-in-circle} and \textsc{torus-in-circle}. The reference PD is always computed on the clean topologically relevant points only.

\paragraph{Protein and dynamical system datasets.}
We use three protein structure point clouds and one dynamical system dataset. For NALCN, we report two official aligned variants: 7WJI and 7SX3-CA. The 7WJI setting follows the TPCC atom filtering, while 7SX3 CA follows the TOPF Figure 3 CA atom preprocessing. For Cys123, we use the 2NU9 CA protein structure aligned with TOPF Figure 4. For all protein datasets, we add $10\%$ uniformly sampled cube outliers around the clean protein coordinates. Clean protein atoms are signal and outliers are non signal. For the Lorenz attractor, we use points sampled from the attractor state space and add $10\%$ cube outliers. The clean attractor points are signal. NALCN and Lorenz are evaluated mainly by Wasserstein distance between 1-dimensional PDs, while Cys123 is evaluated mainly by Wasserstein distance between 2-dimensional PDs.

\paragraph{Grid cell dataset.}
For the large distance-based experiment, we use the grid cell torus dataset. The input is a precomputed correlation distance matrix rather than Euclidean coordinates, so alpha complexes are not applicable. We therefore use Rips complexes for both TRLS scoring and  evaluation. The main grid cell TRLS variant is TRLS-H1, using sparse Rips scoring with sparse parameter $\epsilon=0.99$, scale $0.5$, and eigen budget $R=4$.

\paragraph{TRLS parameters.}
Unless otherwise specified, TRLS uses an RKDE reliability score and an alpha complex built at a heuristic scale. The default scale is selected from local neighbor distances with neighbor parameter $k=8$, weighted quantile $0.90$.  For each target dimension $q$, TRLS constructs the Hodge Laplacian $L_q$ on the selected complex and computes its near zero eigenspace using a sparse Lanczos eigensolver. The default eigen budget is $R=20$, and eigenvectors with eigenvalue below $10^{-8}$ are used to form simplex energy scores. Simplex scores are lifted to point scores, and scores from multiple dimensions are fused by max pooling. The default TRLS uses dimensions $q\in\{1,2\}$ when both loop and void structures are relevant; TRLS-H1 and TRLS-H2 use only $q=1$ and $q=2$, respectively. 

\paragraph{TopoGraph configuration.}
For an input subset of size $m$, we set the DTM and adaptive kNN parameters to the same value:
\begin{equation}
k_{\mathrm{DTM}}=k_{\mathrm{AkNN}}=
\begin{cases}
\min(20,m-1), & m\leq 2000,\\
\lfloor 0.005m\rfloor, & m>2000.
\end{cases}
\end{equation}

The result hyperparameters are set to be the same in \url{https://github.com/tulip-lab/topograph}.

The target number of super nodes is $N_c=\rho m$, where $\rho=0.1$. Each activity vector is centered and normalized to unit norm before applying TopoGraph, so that Euclidean distance represents the correlation distance used by the other methods. Following the official implementation, the first original point in each super node is selected as its landmark.

The proportional rule $k=0.005m$ produces excessively sparse and disconnected graphs for small $m$, preventing graph coarsening from reaching $N_c$. We therefore use $k=20$ for these cases, capped by $m-1$ when necessary. The original proportional setting is retained for larger inputs, where it is stable. If no reliable knee is detected during DTM filtering, all points are retained, since applying an undefined threshold would make the cleaning step invalid.

\paragraph{Other baseline settings.}
We compare with Random, FPS, DA-FPS, PHL-I, and PHL-II. Random samples landmarks uniformly without replacement. FPS uses farthest point sampling with a random initial point. DA-FPS follows density aware farthest point sampling and uses the same sampling density as the other methods. PHL-I and PHL-II follow the persistent homology landmark baseline of Stolz et al. For each point, they compute local persistent homology in a neighborhood $\delta=0.25$ (deafult setting in \citep{stolz2023outlier}) around that point and use the resulting local persistence score for ranking. PHL-I uses all dimensional local persistence with maximum dimension $2$, while PHL-II uses dimension restricted local persistence with maximum dimension $1$. In our implementation, PHL local computations use $8$ parallel workers in runtime experiments. Random, FPS, DA-FPS, and TRLS are repeated when randomness is involved; deterministic PHL scores are computed once.

\paragraph{Setting in robustness experiment}
\label{app:robustness setting}

We evaluate robustness on the \textsc{Torus-in-circle} dataset containing 1,200 torus points, 1,200 ring points, 800 connector points, and 600 points from three solid spheres. The torus and ring form the topologically relevant signal, while the connector, solid spheres, and added outliers are irrelevant. The sampling density is fixed to $\rho=0.2$.

For the noise sweep, Gaussian noise with standard deviation
$\sigma\in\{0,0.005,\ldots,0.08\}$ is added to the base point cloud, while the outlier ratio is fixed to $0.1$. For the outlier sweep, the Gaussian noise is fixed to zero and uniform cube outliers are added at ratios
$p\in\{0,0.025,\ldots,0.30\}$. The main figure reports the range $p\in[0.10,0.25]$. For a target ratio $p$, we append
\begin{equation}
    n_{\mathrm{out}}=
    \operatorname{round}\left(\frac{pN}{1-p}\right)
\end{equation}
outliers to the $N=3,800$ base points.

We compare TRLS with Random, FPS, DA-FPS, PHL-I, PHL-II, and TopoGraph. PHL uses the same neighborhood radius $\delta=0.25$ as in the main experiment. For stochastic methods, landmark selection is repeated five times and we report the mean and standard deviation. Each perturbation level uses one fixed point cloud realization.

Signal fraction is the proportion of selected landmarks belonging to the torus or ring. Topological preservation is measured by the sum of the second-order Wasserstein distances between the selected and clean-signal persistence diagrams in dimensions one and two. All diagrams are computed using the alpha filtration, retaining finite features with persistence at least $0.01$. The reference diagram is always computed from the unperturbed torus and ring.


\section{Proof of Proposition 1}
\label{appendix: proof}

\textbf{Proposition 1.} Under Definition 2, 
the topologically relevant and irrelevant points form a partition of \(X_{\mathcal M}\), i.e., $X_{\mathcal M}=X_{\mathrm{r}}\cup X_{\mathrm{ir}}, X_{\mathrm{r}}\cap X_{\mathrm{ir}}=\varnothing.$

\begin{proof}
By definition, for every \(q\in Q\),
\[
X_{\mathrm{r}}^q = X_{\mathcal M}\setminus X_{\mathrm{ir}}^q.
\]
Therefore, we have
\[
X_{\mathrm{r}}
=
\bigcup_{q\in Q}X_{\mathrm{r}}^q
=
\bigcup_{q\in Q}\left(X_{\mathcal M}\setminus X_{\mathrm{ir}}^q\right).
\]
We now show that
\[
\bigcup_{q\in Q}\left(X_{\mathcal M}\setminus X_{\mathrm{ir}}^q\right)
=
X_{\mathcal M}\setminus \bigcap_{q\in Q}X_{\mathrm{ir}}^q.
\]
For any point \(x\), we have
\[
x\in \bigcup_{q\in Q}\left(X_{\mathcal M}\setminus X_{\mathrm{ir}}^q\right)
\]
if and only if there exists some \(q\in Q\) such that
\[
x\in X_{\mathcal M}
\quad\text{and}\quad
x\notin X_{\mathrm{ir}}^q.
\]
Since \(X_{\mathcal M}\) does not depend on \(q\), this is equivalent to
\[
x\in X_{\mathcal M}
\quad\text{and}\quad
\exists q\in Q,\ x\notin X_{\mathrm{ir}}^q.
\]
The second condition is equivalent to
\[
x\notin \bigcap_{q\in Q}X_{\mathrm{ir}}^q.
\]
Thus, we have
\[
x\in X_{\mathcal M}
\quad\text{and}\quad
x\notin \bigcap_{q\in Q}X_{\mathrm{ir}}^q,
\]
which means
\[
x\in X_{\mathcal M}\setminus \bigcap_{q\in Q}X_{\mathrm{ir}}^q.
\]
Hence, it holds that
\[
X_{\mathrm{r}}
=
X_{\mathcal M}\setminus \bigcap_{q\in Q}X_{\mathrm{ir}}^q.
\]
Using the definition
\[
X_{\mathrm{ir}}=\bigcap_{q\in Q}X_{\mathrm{ir}}^q,
\]
we can obtain
\[
X_{\mathrm{r}}=X_{\mathcal M}\setminus X_{\mathrm{ir}}.
\]
Finally, we have
\[
X_{\mathrm{r}}\cup X_{\mathrm{ir}}
=
\left(X_{\mathcal M}\setminus X_{\mathrm{ir}}\right)\cup X_{\mathrm{ir}}
=
X_{\mathcal M}.
\]
and 
\[
X_{\mathrm{r}}\cap X_{\mathrm{ir}}
=
\left(X_{\mathcal M}\setminus X_{\mathrm{ir}}\right)\cap X_{\mathrm{ir}}
=
\varnothing.
\]
Thus, \(X_{\mathrm{r}}\) and \(X_{\mathrm{ir}}\) form a partition of
\(X_{\mathcal M}\).
\end{proof}

\section{On the use of Large Language Models}
\label{appendix: LLM}

 Large Language Models (LLMs) were utilized in the polishing phase of this paper’s preparation. Specifically, LLMs were employed to optimize linguistic clarity, enhance stylistic coherence, and correct minor grammatical or syntactical inconsistencies. All core intellectual content, including conceptual frameworks, empirical observations, argumentative structure, and citation alignment, was developed, curated, and validated exclusively by the human authors.

 Portions of the code in this project were developed with the assistance of large language models (LLMs), including support for code scaffolding, debugging, and documentation refinement. All LLM-generated or LLM-modified content was manually reviewed, tested, and revised as necessary by the authors. The authors take full responsibility for the correctness and reliability of the final code and the resulting research outcomes.




\end{appendices}


\bibliography{sn-bibliography}

\end{document}